\documentclass{article}
\usepackage{geometry,amsthm,amsmath,amssymb,graphicx,float,enumerate,hyperref}
\usepackage{subfigure}
\usepackage{arydshln}
\usepackage{pgf,tikz}
\usepackage{mathrsfs}
\usepackage[font=small]{caption}

\usepackage{mathpazo}

\definecolor{qqwuqq}{rgb}{0.,0.39215686274509803,0.}
\definecolor{sqsqsq}{rgb}{0.12549019607843137,0.12549019607843137,0.12549019607843137}
\definecolor{qqqqff}{rgb}{0.,0.,1.}
\definecolor{cqcqcq}{rgb}{0.7529411764705882,0.7529411764705882,0.7529411764705882}
\definecolor{eqeqeq}{rgb}{0.8784313725490196,0.8784313725490196,0.8784313725490196}

\usetikzlibrary{arrows}

\swapnumbers
\newtheorem{theorem}{Theorem}[section]
\newtheorem{corollary}[theorem]{Corollary}
\newtheorem{lemma}[theorem]{Lemma}
\newtheorem{proposition}[theorem]{Proposition}

\theoremstyle{definition}

\theoremstyle{remark}
\newtheorem{remark}[theorem]{Remark}
\newtheorem{example}[theorem]{Example}

\DeclareMathOperator{\diag}{diag}

\DeclareMathOperator{\rank}{rank}
\DeclareMathOperator{\tr}{trace}

\title{On the Design of Multi-Dimensional Compactly Supported Parseval Framelets with Directional Characteristics}

\author{N. Karantzas$^1$, N. Atreas$^2$, M.Papadakis$^1$, T. Stavropoulos$^3$}
\date{{\small 1. Dept. of Mathematics, University of Houston, USA} \\
{\small2. Dept. of Electrical and Computer Engineering,\\ Aristotle University of Thessaloniki, GR}\\
{\small 3. Dept. of Mathematics, National and Kapodistrian University of Athens, GR}}

\begin{document}
\maketitle
\abstract 
\noindent In this paper, we propose a new method for the construction of multi-dimensional, wavelet-like families of affine frames, commonly referred to as framelets, with specific directional characteristics, small and compact support in space, directional vanishing moments (DVM), and axial symmetries or anti-symmetries. The framelets we construct arise from readily available refinable functions. The filters defining these framelets have few non-zero coefficients, custom-selected orientations and can act as finite-difference operators. The article has been accepted for publication in Linear algebra and its applications and the corresponding DOI is \url{https://doi.org/10.1016/j.laa.2019.07.028}.

\medskip


\medskip


\section{Introduction}\label{s1}
Multidimensional sparse representations occupy a significant part of the literature on multiscale decompositions. The interest in such representations arises from their ability, at least in theory, to detect singularities along curves of surfaces with some smoothness. However, it is not the first time that such representations are developed for the analysis of 2D and 3D images. From the early years of filter banks and wavelets, image decompositions for compression and analysis have been on the focus of many researchers (e.g., \cite{Daubechies_1992_4532,Vetterli_1995_4637}). The vast majority of those designs was based on tensor product constructs of one-dimensional multiscale decompositions.

However, even in the early 90s, it was realized that such constructs (those mostly in use at the time were real-valued) do not seem to give optimal results, especially on curved boundaries \cite{Kovacevic_1992_4518,Vetterli_1995_4637}. This motivated several researchers to explore non-separable (non-tensor product) designs, e.g., \cite{Ayache_2001_6383,Belogay_1999_58} or other dilation operators, e.g., \cite{Kovacevic_1992_4518} that later led to the quite popular design of beamlets, curvelets and shearlets (\cite{Candes2006,Candes2005,Demanet_2003_4534,Candes_1999_6387,Candes_2004_12186,Labate2005}). The starting point is a refinable function $\phi$; wavelets are then derived using the classical equations involving the low and high-pass filters, generalized as Extension Principles first by Ron and Shen \cite{Ron_1997_4616,Ron_1997_4615}. Stability, compact support, smoothness and vanishing moment orders of the resulting wavelets are derived from properties of the generating refinable function, e.g., \cite{Daubechies_1992_4532}.  

In this paper, we attempt to propose an alternative view on this old problem. Our goal is to combine anisotropy with an  abundance of orientations to mimic those of discrete  curvelets and shearlets (see \cite{Grohs2014} for a comprehensive treatment of the directionality for discrete parabolic molecules). Specifically, we propose a new method to design frame wavelets which combine the advantages of  compactly supported wavelets, namely small support and vanishing moments, but also the directionality and orientability of curvelets and shearlets. One of the key novelties of this work is that we trade classical filter design, formulated as a problem of solving systems of trigonometric polynomial equations in the frequency domain for a much more computationally efficient method based on Singular Value Decomposition (SVD) (Theorems \ref{t2.5} and \ref{t3.2}). This new method is simple and is the key contribution of this work.

Our starting point is a refinable function with compact support, or in other words, a function $\phi\in L_{2}(\mathbb{R}^{s})$ satisfying the following conditions:

\begin{itemize}
\item The Fourier transform $\widehat{\phi}$ is continuous in a neighborhood of the origin and \[\widehat{\phi}(0)=1.\]
\item The $\mathbb{Z}^{s}$-periodic function $\Phi=\sum_{k\in\mathbb{Z}^{s}}|\widehat{\phi}(\cdot+k)|^2$ is in $L_{\infty}(\mathbb{T}^{s})$, the space of all measurable essentially bounded functions on $\mathbb{T}^{s}$. The {\it spectrum} of $\Phi$ is denoted by $\sigma_{\phi}=\{\gamma\in \mathbb{T}^{s}:\Phi(\gamma)\neq 0\}$.
\item $\phi$ is a {\it refinable function}, i.e., $\widehat{\phi}(2\gamma) = H_{0}(\gamma)\widehat{\phi}(\gamma)$ for almost every $\gamma$ and for some $\mathbb{Z}^{s}$-periodic function $H_{0}\in L_{2}(\mathbb{T}^{s})$ called a {\it low-pass filter} or a {\it refinement mask}. 
\end{itemize}

\noindent Next, given a finite natural number $v$, we also consider a vector of refinable functions $\Psi =(\psi_{i})_{i=1}^{v}\in L_{2}^{1 \times v}$ called a {\it multi-wavelet} satisfying $\widehat{\Psi}(2\gamma)=H_{1}(\gamma)\widehat{\phi}(\gamma)$ for a.e. $\gamma \in \mathbb{R}^{s}$ and for another $\mathbb{Z}^{s}$-periodic vector-valued function $H_{1}\in L_{2}^{v \times 1}(\mathbb{T}^{s})$ called a {\it high-pass filter} or a {\it wavelet mask}. We define the dilation and translation operators on $L_{2}(\mathbb{R}^{s})$ by $D_{2}f=2^{s/2}f(2\cdot)$ and $\tau_{k}f = f(
\cdot-k)$, $k\in\mathbb{Z}^{s}$, respectively. For the above selection of the vector $\Psi$ we define its corresponding {\it homogeneous wavelet family} or {\it affine family} $X_{\Psi}$ by 
\[X_{\Psi}=\left\{\psi_{i,j,k}=D_{2}^{j}\tau_{k}\psi_{i}:j\in\mathbb{Z}, k\in \mathbb{Z}^{s}, i=1,...,v\right\}.\]
Additionally, for any $j_{0}\in\mathbb{Z}$, we define the {\it non-homogeneous wavelet family} $X_{\phi,\Psi}^{(j_{0})}$ by
\[X_{\phi,\Psi}^{(j_{0})}=\left\{D_{2}^{j}\tau_{k}\psi_{i}: j\geq j_{0}, k\in\mathbb{Z}^{s}, i=1,\ldots,v\right\}\cup\left\{D_{2}^{j_{0}}\tau_{k}\phi: k\in\mathbb{Z}^{s}\right\}.\]
If there exist two positive constants $C_{1}$ and $C_{2}$, such that the inequality
\[C_{1}\Vert f\Vert_{2}^{2}\leq\sum_{j\in\mathbb{Z}}\sum_{k\in\mathbb{Z}^{s}}\sum_{i=1}^{v}|\langle f,\psi_{i,j,k}\rangle|^{2}\leq C_{2}\Vert f\Vert_{2}^{2}\] 
holds for any $f\in L_{2}$, we say $X_{\Psi}$ is an {\it affine frame} or a {\it homogeneous wavelet frame} for $L_{2}$ and the elements of $\Psi$ are often called {\it framelets}. Here, we sometimes refer to them as frame wavelets. If $C_{1}=C_{2}$, then $X_{\Psi}$ is called a {\it tight wavelet frame} and if $C_{1}=C_{2}=1$, then $X_{\Psi}$ is called a {\it Parseval wavelet frame} or {\it Parseval framelet}. Homogeneous wavelet frames have only theoretical interest. In applications we are more interested in non-homogeneous frames because they model an image decomposition into various fine scales and a coarse residual created by the integer translates of the refinable function. 

Our work is influenced by \cite{Ron_1997_4616} followed by the work of  \cite{DAUBECHIES20031,Han2012,CHUI2002224}. The \emph{Mixed Oblique Extension Principle} which characterizes the pairs of "dual" families of homogeneous and non-homogeneous frames was generalized by \cite{ATREAS201451,Atreas2016} and broadens the applicability of the \emph{Unitary Extension Principle}. Here we focus on UEP, but we believe that our methods can be extended for MOEP. 

Our goal is not to propose new filters and framelets, but to provide a design framework through which one can create ensembles of Parseval framelets defined by sets of high pass finite-length filters, which can be a mix of well-known filters as well as other custom-made ones. Our intent is to make those Parseval framelet ensembles suitable to capture edges, textures and surfaces of singularities with enough sensitivity in preselected orientations. Additionally, the use of compact support promotes sparsity, which is important for many  applications.  In that regard, our gold standard is the  sparsity asymptotics of continuous curvelets and shearlets, e.g., \cite{Candes_2004_12186,Kutyniok_2009_12187}. Both families  achieve this optimal sparsity by continuously increasing the orientation resolution with scale, something our constructs are not meant to do, because they form discrete frames. However, the small compact support of our framelets in space gives them an advantage that curvelets and shearlets lack, because those are compactly supported in frequency, with the notable exception of the compactly supported shearlets developed in \cite{Kittipoom2012}.  Those form only approximately homogeneous frames and their filter length is a multiple of that of our filters.

Looking back in the design of affine wavelets in multidimensions, the vast majority of them are orthonormal or Riesz wavelets defined as tensor products of one-dimensional multiresolution analysis wavelets. Tensor product constructs tend to favor horizontal or vertical image characteristics and even introduce directional filtering variability depending on orientation. This fact was recognized by Kovacevic and Vetterli \cite{Kovacevic_1992_4518}, who attempt to construct the first finite length filters for non-tensor product filter banks.  Notably, different are the non-tensor product constructs of 
\cite{Belogay_1999_58, Ayache_1997_29, ehler05, SANANTOLIN2013201, HAN200343, Han_Jiang_Shen_Zhuang, HAN1997380} which start from a single, compactly supported refinable function whose integer shifts form a Riesz or an orthonormal basis (see \cite{ehler05} for an interesting multidimensional MRA, non-tensor product-design literature review). General dilation matrices and properties such as compact support, decay, smoothness, symmetry and vanishing moments are explored in depth. We remark that all these constructs produce only real-valued wavelets. A nice, alternative way which combines directionality and avoids the preferred filtering orientations of real-valued tensor products is the introduction of complex-valued wavelets and frames  pioneered by Kingsbury \cite{kingsbury_1999_394, Selesnick_2000_4620} and more recently \cite{Han_Zhao, Han_Mo_Zhao}, which also attempt to reproduce the anisotropy of parabolic molecules.

The construction of refinable functions with stable integer shifts is all but an easy task, as the work of Cabrelli \emph{et. al.} \cite{Cabrelli_2000_4507} demonstrates. Therefore, it is quite easier to resort to plain refinable functions whose integer shifts form a Bessel family. In this manuscript, we fully adopt this position which breaks away from the MRA-orthodoxy. As Ron and Shen demonstrated \cite{Ron_1997_4616}, this can be done with the so-called Extension Principles with added benefits, the combination of small filter support with symmetry or antisymmetry. 

An entirely different approach was proposed in \cite{Adelson_1987_3,Simoncelli_1995_606} where a filter-bank precursor of directional atoms was proposed, the steerable pyramids, aiming to define rotationally covariant multiscale transforms. In theory, rotational covariance can be realized by continuous directional transforms such as the Curvelet and Shearlet transforms. For discrete transforms this is not always obviously true or even realizable. Nonetheless, some rotational covariance can be achieved also by directional atoms as in \cite{Candes2006,Candes2005,Candes_1999_4514,Guo_2007_12169}. In this context, the rotational covariance of the representation is important because it makes feature extraction resistant to misclassification of structures due to rotations (e.g., \cite{Papadakis_2009_4602}). With shearlets, rotational covariance is different because different orientations are implemented by powers of the shearing matrices and not by rotations. Results in \cite{Grohs2014} may help elucidate this fact.
At any rate, if frame atoms are directional and orientable (e.g \cite{Candes2006,Candes2005,Candes_1999_4514,Demanet_2003_4534,Candes_1999_6387,Candes_2004_12186,Adelson_1987_3,Labate2005}), then rotational covariance is well-approximated because the induced data transforms can be thought of as good approximations of their continuous counterparts. 

More recently, a very interesting "projection method" has been proposed by B. Han to define framelets with small supports in various orientations \cite{BinHan2017}. We reproduce the main results of \cite{BinHan2017} in Corollary \ref{c2.6}. The difficulty to construct orientable frame atoms with small spatial support motivated us to seek an alternative way to construct multi-scale framelets or, more generally, atoms with this kind of support in space, oriented to have targeted filtering selectivity along a single direction selected by us from a set of several, pre-determined orientations. We can increase the number of those orientations by enlarging the spatial support of the generating refinable function. This construction method as well as the ability to keep the filters short in length are the main contribution of this paper. Furthermore, we can make filter orientation selectivity razor sharp by increasing the support of the refinable function while retaining the remaining desirable properties of the filters. 

Our main objective, the framelet construction method with respect to isotropic dyadic dilations we introduce here, is based on Theorem \ref{t2.5} which bears no similarity with classical wavelet constructions. The refinable functions we use are tensor products of  one-dimensional spline functions, which endows $\Psi$ with axial symmetries, sufficient smoothness and compact support. We are bound to to use refinable functions whose low pass filter coefficients are positive. Surprisingly enough, we show in Section \ref{s4} that the only significance of the choice of the refinable function is limited to the number of its low pass filter coefficients. This is the main reason why we are not interested in expanding our refinable function universe beyond tensor products of $B$-splines.  The essence of our design approach is that framelets $\Psi$ are derived by any high pass filter $H$, as long as $H(0)=0$ and the support of $H$ is contained in the support of the low pass filter (Section \ref{s4}). Of course, there is an associated cost for this procedure because it is rather unlikely that we can construct sets of Parseval Framelets exclusively containing the high pass filters $H$ of our choice. The multi-wavelet $\Psi$ will likely contain other framelets introduced by the process Theorem \ref{t2.5} prescribes, but as we show in Theorem \ref{t3.2}, these auxiliary elements of $\Psi$ may end up having negligible contributions in image reconstructions. 

The framelets we construct have similar properties with parabolic molecules \cite{Grohs2014}, but unlike the latter, the number of their orientations is fixed for all scales. The orientation of parabolic molecules is defined in the frequency domain. This is not suitable for us, since our framelets have compact support in space and are not $C^{\infty}$. In fact, they are less smoother. Directional filter banks, as well as atoms with higher order directional vanishing moments were studied in \cite{Adelson_1987_3,Simoncelli_1995_606,MinhDo05,MinhDo3,MinhDo07,MinhDo4}. All of them are constructed in the frequency domain. One of our novelties is the adaptation of these concepts in the spatial domain. We also provide a characterization of the Directional Vanishing Moment (DVM) orders of wavelets and an algorithmic construction to generate wavelets with up to $N-1$ DVMs. Moreover, we can customize our DVMs to be directed toward a certain orientation which does not have to coincide with the orientation of its wavelet. This helps to increase local sensitivity to wavefronts with the same orientation. Although directionality is a frequently used term in this article, we do not attempt to define it rigorously. In fact, a careful examination of the literature reveals that other authors, who use the term, avoid to do so. We invoke directionality in a descriptive manner in the sense that such directional filters or framelets have pronounced anisotropies in certain orientations, but may also have directional vanishing moments not necessarily aligned with their pronounced orientation or its normal. 

This manuscript is divided in three main sections. In Section \ref{s2}, we begin our discussion with the equations of the UEP, which we use to derive a linear algebra method which transforms the design problem of framelets arising from a refinable function to a problem of designing Parseval frames in finite-dimensional spaces. In Section \ref{s3}, we develop an algorithm which allows to custom-select the orientation and other properties of the filters defining these Parseval framelets in order to achieve high spatial orientation of the resulting high pass filters. Finally, in Section \ref{s4} we show how to include high pass filters of our choice in the high pass filter set defining $\Psi$ and present several typical examples of the filter design strategies we propose based on the methods we develop in the preceding two sections.

\section{The geometry of the proposed construction}\label{s2}

The starting point for our method is that $X_{\Psi}$ is a Parseval framelet for $L_{2}$ if and only if there exists a complex-valued vector function $H_{1}\in L_{2}^{v\times 1}(\mathbb{T}^{s})$, $v>0$, satisfying 
\begin{equation}
\label{1}
\overline{H_{0}(\gamma+q)}H_{0}(\gamma)+H_{1}^{*}(\gamma+q)H_{1}(\gamma)=\delta_{0,q}
\end{equation}
for all $q\in\{0,1/2\}^{s}$ and for almost every $\gamma\in\mathbb{T}^{s}$. Equations \eqref{1}, first presented in \cite{Ron_1997_4616}, are called the Unitary Extension Principle, according to which if the first row of the \emph{modulation matrix}
\[\begin{pmatrix}
H_{0}(\gamma)&H_{1,1}(\gamma)&\cdots&H_{1,v}(\gamma)\\
H_{0}(\gamma+q_{2})&H_{1,1}(\gamma+q_{2})&\cdots&H_{1,v}(\gamma+q_{2})\\
\vdots&\vdots&\ddots&\vdots\\
H_{0}(\gamma+q_{2^{s}-1})&H_{1,1}(\gamma+q_{2^{s}-1})&\cdots&H_{1,v}(\gamma+q_{2^{s}-1})\\
\end{pmatrix}\]
satisfies
\[|H_{0}(\gamma)|^{2}+\sum_{k=1}^{v}|H_{1,k}(\gamma)|^{2}=1\]
for almost every $\gamma\in\mathbb{T}^{s}$, and if it is orthogonal to every other row, then $X_{\Psi}$ forms a Parseval wavelet frame for $L_{2}(\mathbb{R}^{s})$ associated with $\phi$. Since the modulation matrix has $2^{s}$ rows, we observe that we must have $v\geq 2^{s}-1$.

This part of our work explores a sufficient condition for solving the above system of equations, which in essence is a system of polynomial equations with a large number of degrees of freedom and therefore quite hard to solve in closed form and in a way that yields compactly supported wavelets $\psi_{i}$. In what follows, $H_{0}$ is assumed to be a trigonometric polynomial of the form 
\[H_{0}(\gamma)=\sum_{k=1}^{N}a_{n_{k}}e^{2\pi i n_{k}\cdot\gamma}\] 
for $a_{n_{k}}\in\mathbb{R}\setminus\{0\}$,  $N>1$, and $n_{k}\in J\subset\mathbb{Z}^{s}$, i.e., the exponents of the complex exponentials in the representation of such a low-pass filter are characterized by $s$-dimensional vectors with integer components. 
We also have \[H_{0}(0)=1,\] or equivalently $\sum_{k=1}^{N}a_{n_{k}}=1$. We rewrite $H_{0}$ using the factorization 
\[H_{0}=aw\] 
where $a$ is the $1\times N$ vector of coefficients
\[a=(a_{n_{k}})_{k=1}^{N}\] 
and $w\in\mathbb{C}^{N\times 1}$ is the vector-valued function of complex exponentials given by 
\[w(\gamma)=\left(e^{2\pi in_{k}\cdot\gamma}\right)_{k=1}^{N}.\] 
From now on we express the high-pass filter $H_{1}\in L_{2}^{v\times 1}(\mathbb{T}^{s})$ as 
\[H_{1}=Bw\] 
for some $B\in\mathbb{R}^{v\times N}$. Using these expressions for $H_{0}$ and $H_{1}$, we state the main problem this section addresses.

{\bf Problem [$A$]:} \emph{Let $H_{0}=aw$ be a low-pass filter as above. Given a natural number $v\geq 2^{s}-1$, we want to determine (if it exists) a real matrix $B\in\mathbb{R}^{v\times N}$ such that the $v\times 1$ vector-valued function $H_{1}=Bw$ satisfies equation \eqref{1} and so its corresponding family $X_{\Psi}$ forms a Parseval framelet for $L_{2}(\mathbb{R}^{s})$.}

Focusing on Problem [$A$], we consider $\{m_{kt}\}_{k,t=1}^{N}$ to be the elements of the $N\times N$ matrix 
\begin{equation}
\label{2}
M:=a^{T}a+B^{T}B
\end{equation}
and we notice that equation \eqref{1} can now be written as
\begin{align}
\delta_{0,q}&=w^{*}(\gamma+q)(a^{T}a+B^{T}B)w(\gamma)\nonumber\\
&=\sum_{k=1}^{N}m_{kk}e^{-2\pi in_{k}\cdot q}+\sum_{k,t=1,k\neq t}^{N}m_{kt}e^{-2\pi in_{k}\cdot q}e^{2\pi i(n_{t}-n_{k})\cdot \gamma},\label{3}\end{align}
for all $q\in\{0,1/2\}^{s}$ and for almost every $\gamma\in\mathbb{T}^{s}$. The second summand in the right hand side of equation \eqref{3} is a linear combination of not necessarily distinct exponentials. Specifically, the second term may consist of several monomials associated with the same exponential which means that uniqueness of coefficients cannot be directly assumed, unless all terms associated with the same exponential are grouped. This gives rise to a rather complex system of non-linear equations, even in the case where the number of unknown parameters is not large. Equation \eqref{3} implies that Problem [$A$] has a solution if we can find appropriate entries for the matrix $B$ (hence for $M$) such that for all $\gamma\in\mathbb{T}^{s}$ and for all $q\in\{0,1/2\}^{s}$ the following equations are satisfied:
\begin{align}
&\label{4}\sum_{k=1}^{N}m_{kk}e^{-2\pi in_{k}\cdot q}=\delta_{0,q},\\
&\label{5}\sum_{k,t=1,k\neq t}^{N}m_{kt}e^{-2\pi in_{k}\cdot q}e^{2\pi i(n_{t}-n_{k})\cdot\gamma}=0.
\end{align}
We provide insight on the analysis concerning the system of \eqref{4} and \eqref{5} in Example \ref{ex1}, but for the purpose of this work we study the case where $M$ is a diagonal matrix, or in other words, the case where $m_{kt}=0$ for $k\neq t$. The second summand in equation \eqref{3} vanishes for all $\gamma$ and so equation \eqref{5} is always satisfied. However, the hypothesis that $M$ is diagonal imposes the constraint $v\geq N-1$ as the next Lemma indicates. In other words, the number of non-zero Fourier coefficients of the low-pass filter $H_{0}$ affects the dimensionality of the high-pass filter $H_{1}$. 

\begin{lemma}
\label{l2.1}
Let $H_{0}=aw$ be a low-pass filter supported on a bounded set $J$ as above and let $v\geq 2^{s}-1$. If $M=(m_{kt})_{k,t=1}^{N}$ is a diagonal matrix as in equation \eqref{2}, then
\begin{itemize}
\item[(a)] $m_{kk}>0$ for all $k=1,\ldots,N$.
\item[(b)] $v+1\geq N$.
\end{itemize}
\end{lemma}

\begin{proof}
(a) Since all the components of the vector $a$ in the expression of $H_{0}$ are non-zero, and since the $k$-th element in the diagonal of $M$, $m_{kk}$, corresponds to the square of the norm of the $k$-th column vector of $\begin{pmatrix}a\\ B\end{pmatrix}\in\mathbb{R}^{(v+1)\times N}$, we have $m_{kk}>0$. 

\noindent (b) If $v+1<N$, then we would have at least one element of the diagonal of $M$ be equal to zero, which by (a) leads to a  contradiction.
\end{proof}
\noindent In light of Lemma \ref{l2.1}, the pursuit of solutions for Problem [$A$] leads to the following modified formulation:

{\bf Problem [$A'$]:} \emph{Let $H_{0}=aw$ be a low-pass filter with bounded support $J$ such that $H_{0}(0)=1$. Given a natural number 
\[v\geq\max\left\{N-1,2^{s}-1\right\},\] 
we want to determine the real matrices $B\in\mathbb{R}^{v\times N}$ for which the matrix $M$ is diagonal and equation \eqref{4} is satisfied.}

\noindent We now notice that if Problem [$A'$] admits a solution $B$, then $B$ is a solution to Problem [$A$] as well. However, the solutions of Problem [$A$] are not exhausted by the solutions of Problem [$A'$], since solutions of the former arise even when $M$ is not diagonal. With this in mind, from now on we focus on Problem [$A'$] and we show that all its solutions define Parseval frames in finite dimensional spaces, which in turn define high-pass filters $H_{1}$ for homogeneous Parseval wavelet frames $X_{\Psi}$. Lemma \ref{l2.2} helps us get a good picture of the underlying geometry.

\begin{lemma}
\label{l2.2}
Let $\alpha,c\in\mathbb{R}^{1\times N}$, $c\neq 0$ and suppose $D\in\mathbb{R}^{v\times N}$ is such that
\begin{itemize}
\item[(a)] the rows of $\begin{pmatrix}\alpha\\D\end{pmatrix}$ form a Parseval frame for $\mathbb{R}^{N}$.
\item[(b)] $Dc^{T}=0$.
\end{itemize}
Then $\alpha$ and $c$ are collinear vectors.  
\end{lemma}

\begin{proof}
Let $d_{i}$ denote the $i$-th row vector of $D$. Then for $c\in\mathbb{R}^{N}$ our assumptions imply
\[c=\langle \alpha,c\rangle \alpha+\sum_{i=1}^{v}\langle c,d_{i}\rangle d_{i}=\langle \alpha,c\rangle \alpha.\]
Hence, $\alpha$ and $c$ are collinear.
\end{proof}

\begin{lemma}
\label{l2.3}
Let $\alpha\in\mathbb{R}^{1\times N}$ be such that $\Vert \alpha\Vert_{2}=1$. Then for any $v\geq N-1$, there always exists a matrix $D\in\mathbb{R}^{v\times N}$ such that the rows of $\begin{pmatrix}\alpha\\D\end{pmatrix}$ form a Parseval frame for $\mathbb{R}^{N}$.
\end{lemma}

\begin{proof}
We prove the statement by presenting an explicit construction of such a matrix $D$. Suppose $V\in\mathbb{R}^{N\times N}$ is such that its first row vector is equal to $\alpha$ and its columns form an orthonormal set for $\mathbb{R}^{N}$. Therefore, we can write
\[\alpha=e_{1}^{T}V\]
where $e_{1}\in\mathbb{R}^{N\times 1}$ is the first vector of the standard basis for $\mathbb{R}^{N}$. We set 
\[D=\begin{pmatrix}0_{v\times 1}&|&U_{v\times(N-1)}\end{pmatrix}V\]
and assume that the columns of $U$ form an orthonormal set. Such a matrix $U$ exists because $v\geq N-1$. Then
\begin{align*}
\begin{pmatrix}\alpha\\D\end{pmatrix}^{T}\begin{pmatrix}\alpha\\D\end{pmatrix}&=\alpha^{T}\alpha+D^{T}D\\
&=V^{T}\left(e_{1}e_{1}^{T}+\begin{pmatrix}0_{v\times 1}&|&U_{v\times(N-1)}\end{pmatrix}^{T}\begin{pmatrix}0_{v\times 1}&|&U_{v\times(N-1)}\end{pmatrix}\right)V\\
&=V^{T}I_{N}V=I_{N}
\end{align*}
Hence, the columns $\begin{pmatrix}\alpha\\D\end{pmatrix}$ are an orthonormal set of $\mathbb{R}^{N}$ and so the rows of $\begin{pmatrix}\alpha\\D\end{pmatrix}$ form a Parseval frame for $\mathbb{R}^{N}$.
\end{proof}
\begin{remark}
The conclusion of Lemma \ref{l2.3} comes from the fact that if $k\geq N$ and $A$ is a $k\times N$ matrix whose columns form an orthonormal set of vectors in $\mathbb{R}^{N}$, then the rows of $A$ are a Parseval frame for $\mathbb{R}^{N}$. Indeed, let $R=\{r_{1},\ldots,r_{k}\}$ be the rows of $A=[a_{ij}]$. Then for every $x\in\mathbb{R}^{N}$, we have
\begin{align*}
\sum_{i=1}^{k}\left|\left\langle x, r_{i}\right\rangle\right|^{2}&=\sum_{i=1}^{k}\left(\sum_{j=1}^{N}x_{j}a_{ij}\right)^{2}\\
&=\sum_{i=1}^{k}\sum_{j=1}^{N}\sum_{l=1}^{N}x_{j}a_{ij}x_{l}a_{il}\\
&=\sum_{j=1}^{N}\sum_{l=1}^{N}x_{j}x_{l}\sum_{i=1}^{k}a_{ij}a_{il}\\
&=\sum_{j=1}^{N}x_{j}^{2}\\
&=\left\Vert x\right\Vert^{2}.
\end{align*}

\end{remark}
\noindent We are now ready to present the complete solution of Problem [$A'$].

\begin{proposition}
\label{p2.4}
Problem [$A'$] admits a solution if and only if 
\begin{itemize}
\item[(a)] $a_{n_{k}}>0$ for all $k=1,\ldots, N$.
\item[(b)] $H_{0}(q)=\delta_{0,q}$ for $q\in\{0,1/2\}^{s}$.
\end{itemize}
\end{proposition}

\begin{proof}
Based on the statement of Problem [$A'$], let $M=a^{T}a+B^{T}B$ be a diagonal matrix and let $B$ be such that equation \eqref{4} is satisfied. We define the $1\times N$ vector $c=(c_{n_{k}})_{k=1}^{N}$ by
\[c_{n_{k}}=\frac{a_{n_{k}}}{\sqrt{m_{kk}}},\quad k=1,\ldots,N\]
where $a_{n_{k}}$ are the low-pass filter coefficients and we notice that $c$ is well defined since Lemma \ref{l2.1} implies $m_{kk}>0$. Moreover, the low-pass filter condition $H_{0}(0)=1$ gives 
\begin{equation}
\label{6}
\sum_{k=1}^{N}a_{n_{k}}=\sum_{k=1}^{N}c_{n_{k}}\sqrt{m_{kk}}=1,
\end{equation}
while by equation \eqref{4} for $q=0$ we obtain $\sum_{k=1}^{N}m_{kk}=1$, or equivalently, the $1\times N$ vector $m=(\sqrt{m_{kk}})_{k=1}^{N}$ satisfies $\Vert m\Vert_{2}=1$. Next, we note that $M$ is diagonal if and only if there exists a $v\times N$ matrix $D$ such that 
\[B=D\diag(\sqrt{m_{11}},\ldots,\sqrt{m_{NN}})\]
and the rows of $\begin{pmatrix}c\\D\end{pmatrix}\in\mathbb{R}^{(v+1)\times N}$ form a Parseval frame for $\mathbb{R}^{N}$. This implies that for any $\alpha\in\mathbb{R}^{N}$ we have
\begin{equation}
\label{7}
\left\Vert\alpha\right\Vert^{2}_{2}=\left|\left\langle \alpha,c\right\rangle\right|^{2}+\sum_{i=1}^{v}\left|\left\langle \alpha,d_{i}\right\rangle\right|^{2}.
\end{equation}
Applying equation \eqref{7} for $\alpha=m$ and utilizing equation \eqref{6} gives $Dm^{T}=0$. Hence, Lemma \ref{l2.2} implies that $c$ and $m$ are collinear and so $c_{n_{k}}=\lambda\sqrt{m_{kk}}$, or equivalently, $a_{n_{k}}=\lambda m_{kk}$ for some $\lambda\in\mathbb{R}$. By equation \eqref{6} we deduce
\[1=\sum_{k=1}^{N}a_{n_{k}}=\lambda\sum_{k=1}^{N}m_{kk}=\lambda,\]
so $a_{n_{k}}=m_{kk}>0$ for all $k=1,\ldots,N$ by Lemma \ref{l2.1}. Finally, this and equation \eqref{4} also imply $H_{0}(q)=\delta_{0,q}$.

Conversely, if $(a_{n_{k}})_{k=1}^{N}$ is a sequence of positive coefficients, then $c=(\sqrt{a_{n_{k}}})_{k=1}^{N}$ is a well-defined unit vector of $\mathbb{R}^{N}$. For $v\geq N-1$, Lemma \ref{l2.3} implies we can always find a real matrix $D\in\mathbb{R}^{v\times N}$ so that the rows of 
\[\begin{pmatrix}c\\D\end{pmatrix}\] 
form a Parseval frame for $\mathbb{R}^{N}$. Then for $B=D\diag(\sqrt{a_{n_{1}}},\ldots,\sqrt{a_{n_{N}}})$, we have that $c^{T}c+D^{T}D=I_{N}$ is equivalent to $a^{T}a+B^{T}B=\diag(a_{n_{1}},\ldots,a_{n_{N}})$. Hence $M$ is diagonal and $m_{kk}=a_{n_{k}}$. Then 
\[\delta_{0,q}=H_{0}(q)=\sum_{k=1}^{N}a_{n_{k}}e^{-2\pi in_{k}\cdot q}=\sum_{k=1}^{N}m_{kk}e^{-2\pi in_{k}\cdot q}\]
and the proof is complete.
\end{proof}

\noindent A surprising consequence of Proposition \ref{p2.4} is that in order to have a solution to Problem [$A'$], all the Fourier coefficients of the low-pass filter must be positive. Tensor products of spline refinable functions yield low-pass filters satisfying both conditions of Proposition \ref{p2.4}. Next, the first of the main results of this work summarizes the preceding discussion. 

\begin{theorem}\label{t2.5}
Let $H_{0}=aw\in L_{2}(\mathbb{T}^{s})$ be a low-pass filter with positive coefficients supported on a finite set of indices $J$ and suppose $H_{0}(q)=\delta_{0,q}$ for all $q\in\{0,1/2\}^{s}$. Then for $v\geq\max\{N-1,2^{s}-1\}$ and $c=(\sqrt{a_{n_{k}}})_{k=1}^{N}$,
\begin{itemize}
\item[(a)] All solutions of Problem [$A'$] are of the form
\[B=D\diag(\sqrt{a_{n_{1}}},\ldots,\sqrt{a_{n_{N}}})\]
where the rows of $\begin{pmatrix}c\\ D\end{pmatrix}$ form a Parseval frame for $\mathbb{R}^{N}$.
\item[(b)] Such matrices $D$ always exist.
\item[(c)] Any solution $B$ of Problem [$A'$] defines a high-pass filter $H_{1}=Bw$ whose associated family $X_{\Psi}$ forms a homogeneous compactly supported framelet for $L_{2}(\mathbb{R}^{s})$ and therefore is a solution of Problem [$A$].
\end{itemize}
\end{theorem}

\begin{proof}
As we see in the proof of the converse of Proposition \ref{p2.4}, the assumptions imposed on $H_{0}$ guarantee the existence of a diagonal matrix 
\[M=\begin{pmatrix}a\\B\end{pmatrix}^{T}\begin{pmatrix}a\\B\end{pmatrix}\] 
whose entries satisfy 
\[\sum_{k=1}^{N}m_{kk}e^{-2\pi in_{k}\cdot q}=\delta_{0,q}.\]
Now (a) follows from the equivalence between $M$ being a diagonal matrix and the rows of $\begin{pmatrix}c\\D\end{pmatrix}$ forming a Parseval frame for $\mathbb{R}^{N}$. (b) follows directly from Lemma \ref{l2.3}. Lastly, for (c), we have
\begin{align*}
\overline{H_{0}(\gamma+q)}H_{0}(\gamma)+H_{1}^{*}(\gamma+q)H_{1}(\gamma)&=W^{*}(\gamma+q)\begin{pmatrix}a\\B\end{pmatrix}^{T}\begin{pmatrix}a\\B\end{pmatrix}W(\gamma)\\
&=\sum_{k=1}^{N}m_{kk}e^{-2\pi in_{k}\cdot q}\\
&=\delta_{0,q}\end{align*}
Thus $X_{\Psi}$ is a Parseval frame for $L_{2}(\mathbb{R}^{s})$.
\end{proof}

Next, we generalize the construction of directional frame atoms with small spatial support presented in \cite[Theorem 2]{BinHan2017}, where the authors use a ``projection method" to create orientations in the space domain essentially projected from higher dimensional Euclidean spaces to spaces with lower dimensionality. Like ours, their filters act like low order finite difference operators along the orientation of the atom. Here we recreate their main result in a somewhat more general framework, specifically for low-pass filters with positive coefficients satisfying $H_{0}(q)=\delta_{0,q}$ for all $q\in\{0,1/2\}^{s}$. This result was also generalized independently in \cite{Han_Diao}, where the very interesting constructs of Quasi-tight framelets were also first introduced.
 
\begin{corollary}
\label{c2.6}
Let $H_{0}=aw\in L_{2}(\mathbb{T}^{s})$ be a low-pass filter with positive coefficients supported on a finite set $J$ and suppose $H_{0}(q)=\delta_{0,q}$ for all $q\in\{0,1/2\}^{s}$. Then the $N(N-1)/2\times 1$ high-pass filter vector $H_{1}$ with components
\[\sqrt{a_{n_{k}}a_{n_{t}}}\left(-e^{2\pi in_{k}\cdot}+e^{2\pi in_{t}\cdot}\right)\]
for all $k\neq t$ with $k<t$ defines an affine Parseval framelet for $L_{2}(\mathbb{R}^{s})$.
\end{corollary}

\begin{proof}
From the definition of $H_{1}$, we have
\begin{align*}
\begin{pmatrix}H_{0}\\H_{1}\end{pmatrix}(\gamma)&=\begin{pmatrix}a\\B\end{pmatrix}w(\gamma)\\
&=\begin{pmatrix}
a_{n_{1}}&a_{n_{2}}&a_{n_{3}}&\cdots&a_{n_{N-1}}&a_{n_{N}}\\
\hdashline[2pt/2pt]
-\sqrt{a_{n_{1}}a_{n_{2}}}&\sqrt{a_{n_{1}}a_{n_{2}}}&0&\cdots&0&0\\
-\sqrt{a_{n_{1}}a_{n_{3}}}&0&\sqrt{a_{n_{1}}a_{n_{3}}}&\cdots&0&0\\
&&&\ddots&&\\
-\sqrt{a_{n_{1}}a_{n_{N}}}&0&0&\cdots&0&\sqrt{a_{n_{1}}a_{n_{N}}}\\
\hdashline[2pt/2pt]
0&-\sqrt{a_{n_{2}}a_{n_{3}}}&\sqrt{a_{n_{2}}a_{n_{3}}}&\cdots&0&0\\
&&&\ddots&&\\
0&-\sqrt{a_{n_{2}}a_{n_{N}}}&0&\cdots&0&\sqrt{a_{n_{2}}a_{n_{N}}}\\
\hdashline[2pt/2pt]
&&&\vdots&&\\
0&0&0&\cdots&-\sqrt{a_{n_{N-1}}a_{n_{N}}}&\sqrt{a_{n_{N-1}}a_{n_{N}}}\\
\end{pmatrix}\begin{pmatrix}e^{2\pi i n_{1}\cdot\gamma}\\\vdots\\e^{2\pi i n_{N} \cdot\gamma} \end{pmatrix}
\end{align*}
for $\gamma\in\mathbb{T}^{s}$. Essentially, the rows of $B$ are generated from all the possible permutations of non-zero column pairs. This implies that $M=a^{T}a+B^{T}B$ is a diagonal matrix since the columns of $\begin{pmatrix}a\\B\end{pmatrix}$ form an orthogonal set of $N$ vectors in $\mathbb{R}^{N}$. Moreover, computing the norm of the $k$-th column of $\begin{pmatrix}a\\B\end{pmatrix}$ gives
\[a_{n_{k}}a_{n_{1}}+a_{n_{k}}a_{n_{2}}+\ldots+a_{n_{k}}^{2}+\ldots+a_{n_{k}}a_{n_{N}}=a_{n_{k}}\sum_{i=1}^{N}a_{n_{i}}=a_{n_{k}},\]
for all $k=1,\ldots, N$. Therefore, $M=\diag(a)$ and $B$ is a solution of Problem [$A'$]. The result follows by Theorem \ref{t2.5}.
\end{proof}

\section{Wavelets with directional vanishing moments and customizable filters.}\label{s3}

The core message of Section \ref{s2} is that under the assumptions of Theorem \ref{t2.5}, one can construct affine Parseval framelets for $L_{2}(\mathbb{R}^{s})$ arising from a refinable function by constructing Parseval frames for $\mathbb{R}^{N}$. This theorem, not only allows us to translate the difficult problem of solving the system of equations of the UEP into the much more algorithmically tractable problem of designing Parseval frames in finite dimensions, but furthermore enables us to custom-shape the filters defining the sought framelets. For example, sparse filters, edge detection filters, filters inducing wavelets with a high order of vanishing moments etc., are some of the high-pass filter families we know produce informative results in a variety of applications.

Our goal here is to propose a theoretical framework that enables us to hand-pick the high-pass filters that define a Parseval framelet. We can also impose certain directional vanishing moments to increase their sensitivity to singularities in application-specific targeted orientations. These design choices, although not the only realizable ones, drive the filter constructs in Section \ref{s4}. The key tool is Theorem \ref{t2.5}, which dictates that the matrix entries of the filters $h_{1,i}$ are determined by the rows of the sub-matrix $D$ of 
 \[\begin{pmatrix}c\\D\end{pmatrix}\in\mathbb{R}^{(v+1)\times N}, \quad v\geq N-1,\] 
whose rows form a Parseval frame for $\mathbb{R}^{N}$, and $c$ is a given unit norm $1\times N$ vector with positive components defined by the Fourier coefficients of $H_{0}$.  

Customizing filters that define affine multi-dimensional Parseval frames and/or selecting the number and direction of their vanishing moments is not a straightforward task. It requires the development of a number of tools which guarantee that in every Parseval frame filter ensemble we create, we maximize the number of filters with those desirable properties.  Each such filter set may have to contain some filters acting as a complement to the set of filters with pre-designed properties in order to derive a Parseval frame. A significant amount of this section is devoted to making their contributions and their number as small as possible (Theorem \ref{t3.2}). In order to achieve these goals, we first need to develop certain filter design tools utilizing Theorem \ref{t2.5}.
\begin{itemize}
\item[(i)] We begin by presenting a sufficient condition for pre-determining $L$ rows of $D$, or a sub-matrix $D_{1}\in\mathbb{R}^{L\times N}$ whose rows are orthogonal to $c$ so that there exist appropriate matrices $D_{2}$ for which the rows of 
\[\begin{pmatrix}c\\ D_{1}\\ D_{2}\end{pmatrix}\]
form a Parseval frame for $\mathbb{R}^{N}$ [Lemma \ref{l3.1}]. The sub-matrix $D_2$ determines the filters acting as a complement to the set of customized filters defined by $D_1$.
\item[(ii)] Next, we seek a technique to optimize the rows of $D_{1}$ to control redundancy and simultaneously minimize the reconstruction error when we choose to omit the framelets $\psi_{i}$ resulting from $D_{2}$ [Theorem \ref{t3.2}]. The algorithm implementing (i) and (ii) can be found at the beginning of Section \ref{s4}. 
\item[(iii)] Finally, we give a characterization of the directional vanishing moment orders (DVM) of framelets, but also how one can explicitly construct wavelets with up to $N-1$ DVM.
\end{itemize}

\noindent The next Lemma addresses (i). In this setting, the affine framelets induced by the rows of $D_{1}$ are pre-designed but it is not necessary that they form an affine frame for $L_{2}(\mathbb{R}^{s})$. From now on we use the notation
\[Q:=\begin{pmatrix}c\\ D_{1}\end{pmatrix}.\] 

\begin{lemma}
\label{l3.1}
Let $D_{1}$ be a fixed $L\times N$ matrix with rows orthogonal to $c$. If the singular values of $Q$ satisfy $\sigma_{i}\leq 1$ for all $i=1,\ldots,L+1$, then there exists an $N\times N$ matrix $D_{2}$ such that the rows of 
\[\begin{pmatrix}Q\\ D_{2}\end{pmatrix}=\begin{pmatrix}c\\ D_{1}\\D_{2}\end{pmatrix}\]
form a Parseval frame for $\mathbb{R}^{N}$. In this case, the Parseval frame consists of $v=L+N+1$ vectors in $\mathbb{R}^{N}$.
\end{lemma}

\begin{proof}
We prove the case where $L+1\leq N$. Using Singular Value Decomposition (SVD), we have $Q=U\Sigma_{1}V^{T}$ for $U\in\mathbb{R}^{(L+1)\times (L+1)}$ and $V\in\mathbb{R}^{N\times N}$ unitary matrices and 
\[\Sigma_{1}=\begin{pmatrix}\diag(\sigma_{1},\ldots,\sigma_{L+1})&|&{\bf 0}_{(L+1)\times (N-L-1)}\end{pmatrix}\in\mathbb{R}^{(L+1)\times N}.\]
Now let $D_{2}=\Sigma_{2}V^{T}\in\mathbb{R}^{N\times N}$ with
\[\Sigma_{2}=\diag\left(\sqrt{1-\sigma_{1}^{2}},\ldots,\sqrt{1-\sigma_{L+1}^{2}},1,\ldots,1\right)\in\mathbb{R}^{N\times N}.\]
This gives
\[Q^{T}Q+D_{2}^{T}D_{2}=V(\Sigma_{1}^{T}\Sigma_{1}+\Sigma_{2}^{T}\Sigma_{2})V^{T}=VI_{N}V^{T}=I_{N}.\]
The case $L+1>N$ is similar and the proof is omitted.
\end{proof}

We remark that the number of non-zero singular values of $Q$ is directly linked to the total number $v$ of high-pass filters. The larger the number of singular values equal to $1$, the smaller the number of rows of $\Sigma_{2}$ is going to be, thus providing us with a tool to control the overall redundancy of the affine family $X_{\Psi}$. 

However, this is not the only notable aspect of this construction. All singular values $\sigma_{1}\geq\sigma_{2}\geq\cdots\geq\sigma_{L+1}$ come from the pre-designed filters induced by $D_{1}$. If $\sigma_{i}=1$ for $i=1,\ldots,L+1$, then whatever complementary filters we add using $D_{2}$ can be considered as the only part of the framelet construction over which we have no control, for it is determined by $V^{T}$. This observation leads us to consider (ii), the second point mentioned in the beginning of this section. 

One way to control the $D_{2}$-contributions is to eliminate the chance of introducing zeros as singular values, or in other words, by ensuring that $\rank(Q)=N$. As we will see in Theorem \ref{t3.2}, this can be done in a way that keeps the resulting singular values $\sigma_{i}$ as close to $1$ as possible. Nevertheless, this is one aspect of the $D_{2}$-construction we do not control.

The next theorem shows there exist matrices $D_{1}$ for which we can jointly maximize all singular values of $Q$ under the constraint $\sigma_{\max}(Q)\leq 1$. Moreover, provided that $\rank(Q)=N$, we want to see how accurate an approximation of an $L_{2}$ function $f$ one can obtain when disregarding the completion matrix $D_{2}$. For this, recall that if $\Psi=(\psi_{1},\ldots,\psi_{v})$ is a multi-wavelet whose corresponding affine family $X_{\phi,\Psi}^{0}$ forms a Parseval frame for $L_{2}(\mathbb{R}^{s})$, then the Calderon Condition states
\[|\widehat{\phi}(\gamma)|^{2}+\sum_{j=0}^{\infty}\sum_{i=1}^{v}\left|\widehat{\psi}_{i}\left(\frac{\gamma}{2^{j}}\right)\right|^{2}=1.\] 
We define
\[E:=1-\sum_{j=0}^{\infty}\sum_{i=1}^{L}\left|\widehat{\psi}_{i}\left(\frac{\gamma}{2^{j}}\right)\right|^{2}-|\widehat{\phi}(\gamma)|^{2}=\sum_{j=0}^{\infty}\sum_{i=L+1}^{v}\left|\widehat{\psi}_{i}\left(\frac{\gamma}{2^{j}}\right)\right|^{2},\]
as well as the reconstruction error of $f$
\begin{align*}
E(f):&=\Vert f\Vert^{2}_{L_{2}}-\sum_{j=0}^{\infty}\sum_{k\in\mathbb{Z}^{s}}\sum_{i=1}^{L}|\langle f,\psi_{i,j,k}\rangle|^{2}-\sum_{k\in\mathbb{Z}^{s}}|\langle f,T_{k}\phi\rangle|^{2}\\
&=\sum_{j=0}^{\infty}\sum_{k\in\mathbb{Z}^{s}}\sum_{i=L+1}^{v}|\langle f,\psi_{i,j,k}\rangle|^{2}.\end{align*}
We seek to establish a connection between the reconstruction error $E(f)$ and the simultaneously maximized singular values of $Q$.
\begin{theorem}
\label{t3.2}
\begin{itemize}
\item[(a)] Let $c$ be a $1\times N$ vector such that $\Vert c\Vert_{2}=1$ and suppose the rows of $D_{1}$, $\{d_{i}\}_{i=1}^{L}$, satisfy 
\[d_{i}c^{T}=0\] 
for all $i$. For $\lambda\in\mathbb{R}^{L}$, we define
$Q(\lambda):=\begin{pmatrix}c\\\diag(\lambda)D_{1}\end{pmatrix}$
and 
\[f_{c}(\lambda):=\tr\left(Q^{T}(\lambda)Q(\lambda)\right).\] 
Then the problem
\[\mathcal{P}:\begin{cases}\max f_{c}(\lambda)\\\text{subject to }\left\Vert Q^{T}(\lambda)Q(\lambda)\right\Vert\leq 1\end{cases}\] admits a solution.
\item[(b)] Let $\widetilde{\lambda}\in\mathbb{R}^{L}$ be a solution of problem $\mathcal{P}$ and let $\widetilde{D}_{1}=\diag\left(\widetilde{\lambda}\right)D_{1}\in\mathbb{R}^{L\times N}$ be such that $\rank(Q)=N$. Then 
\[E(f)\leq\sigma\Vert f\Vert_{L_{2}}^{2}\]
where $\sigma:=1-\sigma^{2}_{N}$ and the truncated non-homogeneous affine wavelet family 
\[\{D_{2}^{j}T_{k}\psi_{i}: j\in\mathbb{Z}, k\in\mathbb{Z}^{s}, i=1,\ldots,L\}\cup\{T_{k}\phi:k\in\mathbb{Z}^{s}\}\] 
is a frame with lower frame bound $\sigma^{2}_{N}$ and upper frame bound $1$.
\end{itemize}
\end{theorem}

\begin{proof}
(a) We define $\Gamma=\left\{\lambda\in\mathbb{R}^{L}: \left\Vert Q^{T}(\lambda)Q(\lambda)\right\Vert\leq 1\right\}$ and notice that for any $D_{1}\in\mathbb{R}^{L\times N}$ with rows in the orthogonal complement of $c$, if $\lambda\in\Gamma$, then 
\[\sigma_{\max}\left(Q(\lambda)\right)\leq 1.\] 
Moreover, $\Gamma$ is non-empty since $0_{L}\in \Gamma$, but also bounded. Now for a sequence $(\lambda_{n})_{n\in\mathbb{N}}\subset \Gamma$ such that $\lambda_{n}\to\lambda_{0}$, we have
\begin{align*}
\Vert Q^{T}(\lambda_{n})Q(\lambda_{n})-Q^{T}(\lambda_{0})Q(\lambda_{0})\Vert&=\Vert D_{1}^{T}\left(\diag(\lambda_{n})^{2}-\diag(\lambda_{0})^{2}\right)D_{1}\Vert\\
&\leq\Vert D_{1}\Vert^{2}\left\Vert\diag\left(\lambda_{n}^{2}-\lambda_{0}^{2}\right)\right\Vert\to 0
\end{align*}
as $n\to\infty$ and so $\Gamma$ is also closed. The result follows by the continuity of the trace function $f_{c}$.

\noindent (b) Since the rows of $\widetilde{D}_{1}$ are orthogonal to $c$ and since $\Vert c\Vert_{2}=1$, we have $\sigma_{1}=1$. Then by applying Lemma \ref{l3.1} to $\widetilde{D_{1}}$, we have 
\[\Sigma_{1}^{T}\Sigma_{1}=\diag(1,\sigma_{2}^{2},\ldots,\sigma_{N}^{2})\] 
and 
\[\Sigma_{2}^{T}\Sigma_{2}=(0,1-\sigma_{2}^{2},\ldots,1-\sigma_{N}^{2}),\]
where $\Sigma_{1}$ and $\Sigma_{2}$ are defined as in Lemma \ref{l3.1}. First, we claim
\begin{equation}\label{8}\sum_{i=L+1}^{v}|H_{i,1}(\gamma)|^{2}\leq(1-\sigma_{N}^{2})\sum_{i=1}^{v}|H_{i,1}(\gamma)|^{2}.\end{equation}
Indeed, since $\sigma_{N}^{2}-\sigma_{i}^{2}\leq 0$ for all $i=1,\ldots,N$, we notice that the matrix 
\begin{align*}\mathcal{S}:&=\Sigma_{2}^{T}\Sigma_{2}-(1-\sigma_{N}^{2})\left(\Sigma_{1}^{T}\Sigma_{1}+\Sigma_{2}^{T}\Sigma_{2}-\diag(1,0,\ldots,0)\right)\\
&=\diag(0,-\sigma_{2}^{2}+\sigma_{N}^{2},\ldots,-\sigma_{N-1}^{2}+\sigma_{N}^{2},0)
\end{align*}
is negative semi-definite. Hence
\[w^{*}(\gamma)\left(M^{1/2}\right)^{T}V\mathcal{S}V^{T}M^{1/2}w(\gamma)=\sum_{i=L+1}^{v}|H_{i,1}(\gamma)|^{2}-(1-\sigma_{N}^{2})\sum_{i=1}^{v}|H_{i,1}(\gamma)|^{2}\leq0\]
for $M^{1/2}=\diag(\sqrt{a_{n_{1}}},\ldots,\sqrt{a_{n_{N}}})$ and for almost every $\gamma\in\mathbb{T}^{s}$. Next, let $\theta_{j}:\mathbb{T}^{s}\to\mathbb{C}$ be given by
\[\theta_{j}(\cdot)=\prod_{k=0}^{j-1}H_{0}(2^{j-1-k}\cdot)H_{1}(2^{j}\cdot)\]
for any $j\geq0$. Recall that the Fundamental Function $\Theta:\mathbb{T}^{s}\to\mathbb{R}^{+}$ associated with the family $X_{\Psi}$ is given by 
\[\Theta(\cdot)=\sum_{j=0}^{\infty}|\theta_{j}(\cdot)|^{2}\] 
and recall that, \cite{DAUBECHIES20031,CHUI2002224}, for almost every $\gamma\in\mathbb{T}^{s}$ we have
\begin{equation}\label{9}\lim_{j\to\infty}\Theta\left(\frac{\gamma}{2^{j}}\right)=1.\end{equation}
We begin by considering the error of approximation for two scales of resolution. Specifically, using \eqref{8} and the definition of the Fundamental function above, we have
\begin{align*}
\sum_{i=L+1}^{v}\sum_{j=0}^{1}\left|\widehat{\psi}_{i}\left(\frac{\gamma}{2^{j}}\right)\right|^{2}&=\sum_{i=L+1}^{v}\left(\left|\widehat{\psi}_{i}(\gamma)\right|^{2}+\left|\widehat{\psi}_{i}\left(\frac{\gamma}{2}\right)\right|^{2}\right)\\
&\leq\sigma\sum_{i=1}^{v}\left(\left|H_{1,i}\left(\frac{\gamma}{2}\right)\right|^{2}\left|\widehat{\phi}\left(\frac{\gamma}{2}\right)\right|^{2}+\left|H_{1,i}\left(\frac{\gamma}{4}\right)\right|^{2}\left|\widehat{\phi}\left(\frac{\gamma}{4}\right)\right|^{2}\right)\\
&=\sigma\left(\left|H_{1}\left(\frac{\gamma}{2}\right)\right|^{2}\left|H_{0}\left(\frac{\gamma}{4}\right)\right|^{2}+\left|H_{1}\left(\frac{\gamma}{4}\right)\right|^{2}\right)\left|\widehat{\phi}\left(\frac{\gamma}{4}\right)\right|^{2}\\
&=\sigma\sum_{j=0}^{1}\left|\theta_{j}\left(\frac{\gamma}{4}\right)\right|^{2}\left|\widehat{\phi}\left(\frac{\gamma}{4}\right)\right|^{2}\\
&\leq\sigma\Theta\left(\frac{\gamma}{4}\right)\left|\widehat{\phi}\left(\frac{\gamma}{4}\right)\right|^{2}
\end{align*}
for almost every $\gamma\in\mathbb{T}^{s}$. Hence if $j_{0}\in\mathbb{N}$, proceeding inductively using the same technique yields
\begin{align*}\sum_{i=L+1}^{v}\sum_{j=0}^{j_{0}}\left|\widehat{\psi}_{i}\left(\frac{\gamma}{2^{j}}\right)\right|^{2}&\leq\sigma\sum_{j=0}^{j_{0}}\left|\theta_{j}\left(\frac{\gamma}{2^{j_{0}+1}}\right)\right|^{2}\left|\widehat{\phi}\left(\frac{\gamma}{2^{j_{0}+1}}\right)\right|^{2}\\
&\leq\sigma\Theta\left(\frac{\gamma}{2^{j_{0}+1}}\right)\left|\widehat{\phi}\left(\frac{\gamma}{2^{j_{0}+1}}\right)\right|^{2}
\end{align*}
Finally, using \eqref{9} and $\widehat{\phi}(0)=1$ and by letting $j_{0}$ tend to infinity we obtain $E\leq\sigma$. The result follows from Theorem 3.2 of \cite{Hernandez_1996_4563} for Parseval frames.
\end{proof}

\subsubsection*{A characterization of Directional Vanishing Moments (DVM)}

Recall that for a given unit vector $\beta\in\mathbb{R}^{s}$, we say a compactly supported wavelet $\psi$ has $n$ vanishing moments in the direction of $\beta$ if 
\[D^{r}_{\beta}\widehat{\psi}(0)=0\]
for all $r=0,1,\ldots,n-1$, where $D_{\beta}^{r}$ represents the $r$-th order directional derivative in the direction of $\beta$. A routine calculation shows
\[D_{\beta}^{r}\widehat{f}(0)=\mathcal{F}\left(\left(-2\pi i(x\cdot\beta)\right)^{r}f(x)\right)(0)\]
for every compactly supported $f\in L_{1}$, where $\mathcal{F}$ denotes the Fourier transform. The previous equation shows that DVM act just like regular moments, primarily in the direction of $\beta$. As in the one-dimensional case, the number of directional vanishing moments of a wavelet $\psi$ is expected to affect the rate of decay of the frame coefficients with respect to the scale $j$ at various directions at any point, especially at points of singularity. We illustrate this effect with Figure \ref{fdvm} below. Specifically, we consider a cubic polynomial image and the high-pass filter   
\[h=\begin{pmatrix}0.1655&-0.2372&0.0718\\
-0.0073&0.0146&-0.0073\\
-0.0207&0.0414&-0.0207\end{pmatrix}\]
corresponding to a wavelet with four DVM in the direction of $(0,1)$ and we notice that $2D$ convolution with $h$ produces an output with no edges.

\begin{figure}[H]
\centering
\includegraphics[width=0.3\textwidth,height=0.25\textwidth]{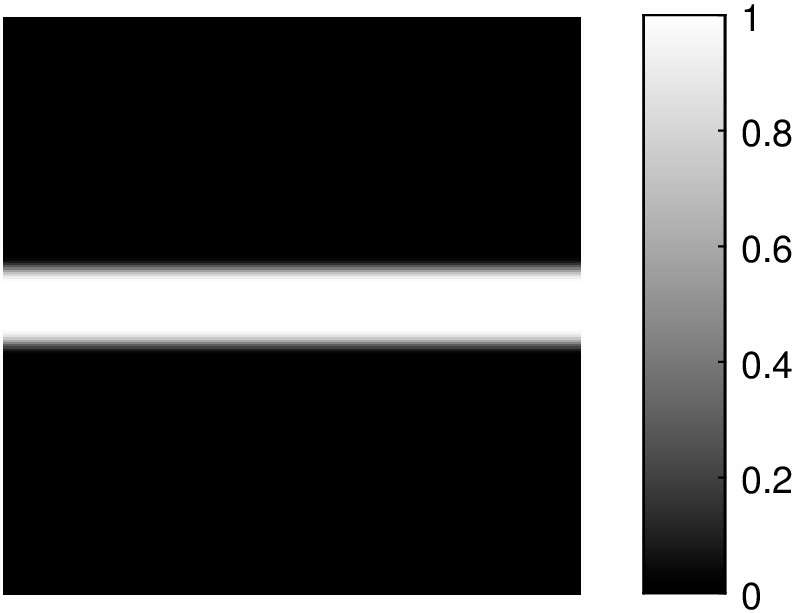}\quad
\includegraphics[width=0.3\textwidth,height=0.25\textwidth]{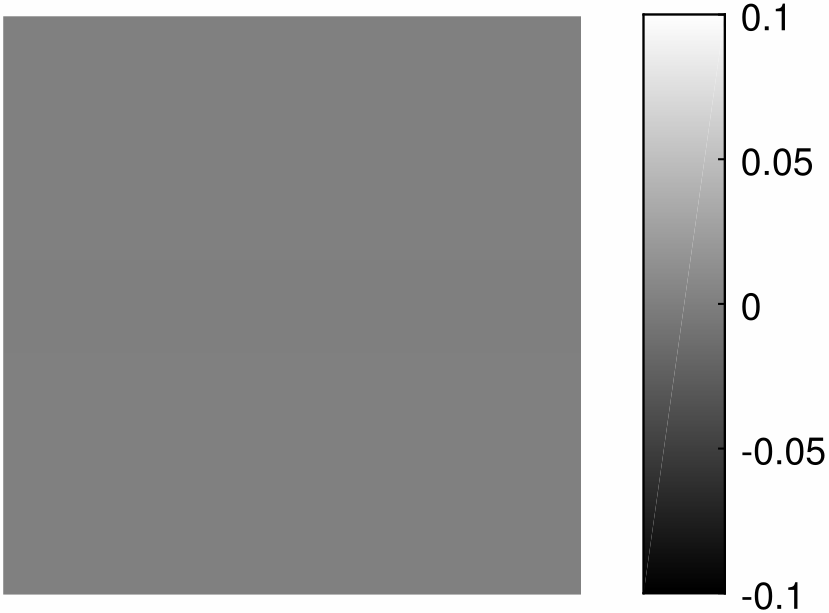}
\caption{Left: Cubic spline interpolation of binary image containing a single horizontal non-zero band. Values in this image are constant in the direction of $(0,1)$. Right: $2D$ convolution with $h$ defining a wavelet with $4$ directional vanishing moments in the direction of $(0,1)$. As expected, the lower polynomial degree of the intensity profile of the left panel relative to the number of DVM of $h$ parallel to $(0,1)$ practically flatten the cubic spline bump in the middle of the left panel.}
\label{fdvm}
\end{figure}

\noindent Next, assuming $B\in\mathbb{R}^{v\times N}$ is a solution to Problem [$A'$], we translate the DVM orders of a wavelet $\psi_{i}$ into certain geometric conditions in $\mathbb{R}^{N}$ via the following characterization:
\begin{proposition}
\label{p3.3}
Let $\beta\in\mathbb{R}^{s}$ and $\Psi=\left(\psi_{i}\right)_{i=1}^{v}$ be a multi-wavelet arising from a matrix $D$ as described in Theorem \ref{t2.5}. Then if $d_{i}$ denotes the $i$-th row vector of $D$, a given wavelet $\psi_{i}$ has $n$ vanishing moments in the direction of $\beta$ if and only if 
\[cZ^{r}d_{i}^{T}=0\]
for all $r=0,1,\ldots,n-1$ and for $Z:=\diag(\beta\cdot n_{1},\ldots,\beta\cdot n_{N})$.
\end{proposition}
\begin{proof}
Since the multi-wavelet $\Psi$ satisfies the two-scale equation $\widehat{\Psi}(2\cdot)=H_{1}(\cdot)\widehat{\phi}(\cdot)$, we infer that $\psi_{i}$ has $n$ vanishing moments in the direction of $\beta$ if and only if $H_{1,i}$ has $n$ vanishing moments in the direction of $\beta$, where $H_{1,i}$ denotes the $i$-th component of $H_{1}$. Next, using $D_{\beta}^{2}$ to denote the second order directional derivative in the direction of $\beta$, we have
\begin{align*}
D_{\beta}^{2}\left(H_{1,i}(\gamma)\right)&=D_{\beta}\left(D_{\beta}\left(H_{1,i}(\gamma)\right)\right)\\
&=D_{\beta}\left(\left(\nabla\sum_{n_{k}\in J}\sqrt{a_{n_{k}}}d_{i,k}e^{2\pi in_{k}\cdot\gamma}\right)\cdot\beta\right)\\
&=2\pi i D_{\beta}\left(\sum_{n_{k}\in J}\sqrt{a_{n_{k}}}d_{i,k}e^{2\pi in_{k}\cdot\gamma}(\beta\cdot n_{k})\right)\\
&=(2\pi i)^{2}\sum_{n_{k}\in J}\sqrt{a_{n_{k}}}d_{i,k}e^{2\pi in_{k}\cdot\gamma}(\beta\cdot n_{k})^{2}.
\end{align*}
Proceeding inductively we find that the $r$-th order directional derivative in the direction of $\beta$ is given by
\[D_{\beta}^{r}\left(H_{1,i}(\gamma)\right)=(2\pi i)^{r}\sum_{n_{k}\in J}\sqrt{a_{n_{k}}}d_{i,k}e^{2\pi in_{k}\cdot\gamma}(\beta\cdot n_{k})^{r}.\]
Therefore, a given wavelet $\psi_{i}$ has $n$ vanishing moments in the direction of $\beta$ if and only if
\[\sum_{n_{k}\in J}\sqrt{a_{n_{k}}}d_{i,k}(\beta\cdot n_{k})^{r}=0\]
for all $r=0,1,\ldots,n-1$, or equivalently if and only if
\[cZ^{r}d_{i}^{T}=0\]
for all $r=0,1,\ldots,n-1$.
\end{proof} 
Next, based on Proposition \ref{p3.3}, we claim that for a given set of low-pass filter polynomial exponents $\{n_{k}\}_{k=1}^{N}$, there exist uncountably many direction vectors for which one can construct wavelets with $N-1$ DVM inducing solutions to Problem [$A'$]. The following proposition supports this claim.

\begin{proposition}\label{p3.4}
There exists a unit vector $\beta\in\mathbb{R}^{s}$ and a vector $d\in\mathbb{R}^{N}$ such that the high-pass filter with coefficients $(\sqrt{a_{n_{k}}}d_{k})_{k=1}^{N}$ induces a wavelet with $N-1$ vanishing moments in the direction of $\beta$.  
\end{proposition}
\begin{proof}
First, we claim that there always exists a vector $\beta\in\mathbb{R}^{s}$ such that all dot products 
\[\beta\cdot n_{k}, \quad k=1,\ldots, N\] 
are distinct. Equivalently, one can always find a $\beta$ such that $(n_{k}-n_{t})\cdot \beta\neq 0$ for all $k\neq t$. Indeed, to not have $(n_{k}-n_{t})\cdot\beta=0$ for some $\beta$ and for all $k\neq t$, we have to exclude $\binom{N}{2}$ hyperplanes from $\mathbb{R}^{s}$. However, by Baire's Category Theorem, $\mathbb{R}^{s}$ is not the union of a finite number of hyperplanes and hence uncountably many such $\beta$ vectors exist. Next, for such a $\beta\in\mathbb{R}^{s}$ we consider the $N\times N$ Vandermonde matrix
\[\mathcal{V}=\begin{pmatrix}1&\cdots&1\\
n_{1}\cdot\beta&\cdots&n_{N}\cdot\beta\\
(n_{1}\cdot\beta)^{2}&\cdots&(n_{N}\cdot\beta)^{2}\\
\vdots&\ddots&\vdots\\
(n_{1}\cdot\beta)^{N-1}&\cdots&(n_{N}\cdot\beta)^{N-1}
\end{pmatrix}\]
for which $\det(\mathcal{V})\neq0$, since all $\beta\cdot n_{k}$, $k=1,\ldots,N$ are distinct. Moreover, the matrix 
\[R:=\mathcal{V}\diag(\sqrt{a_{n_{1}}},\ldots,\sqrt{a_{n_{N}}})=\begin{pmatrix}\sqrt{a_{n_{1}}}&\cdots&\sqrt{a_{n_{N}}}\\
(n_{1}\cdot\beta)\sqrt{a_{n_{1}}}&\cdots&(n_{N}\cdot\beta)\sqrt{a_{n_{N}}}\\
(n_{1}\cdot\beta)^{2}\sqrt{a_{n_{1}}}&\cdots&(n_{N}\cdot\beta)^{2}\sqrt{a_{n_{N}}}\\
\vdots&\ddots&\vdots\\
(n_{1}\cdot\beta)^{N-1}\sqrt{a_{n_{1}}}&\cdots&(n_{N}\cdot\beta)^{N-1}\sqrt{a_{n_{N}}}
\end{pmatrix}\] 
is invertible, since $a_{n_{k}}\neq0$ and so the last column vector of $R^{-1}$ is orthogonal to all first $N-1$ rows of $R$. Therefore, by Proposition \ref{p3.3}, choosing $d$ to be the last column vector of $R^{-1}$ and applying Theorem \ref{t3.2}(a) implies that the corresponding wavelet $\psi$ has $N-1$ vanishing moments in the direction of $\beta$.
\end{proof}
\begin{remark}
Although we cannot expect the order of directional vanishing moments to exceed $N-1$, the previous proposition shows that there are uncountably many direction vectors $\beta$ for which this order of moments is realized. 
\end{remark}

\section{Examples}\label{s4}

As indicated in Sections \ref{s2} and \ref{s3}, the purpose of this work is to develop techniques to handcraft affine Parseval framelet sets, or at least handcraft the part of them which most significantly contributes to multidimensional image reconstructions.
In this section, we propose a four-step algorithmic process via which, for any high-pass filter 
\[H(\cdot)=\left(H_{1}(\cdot),\ldots,H_{L}(\cdot)\right)^{T}\in L_{2}^{L\times 1}(\mathbb{T}^{s})\] with components $H_{i}(\cdot)=\sum_{k=1}^{N}b_{n_{k}}^{i}e^{2\pi in_{k}\cdot}$, $i=1,\ldots,L$, one can force a Parseval framelet for $L_{2}(\mathbb{R}^{s})$ to comprise wavelets $\psi_{i}$ with corresponding high-pass filters (up to scalar multiplications). Using this algorithm, we construct classes of representative examples of explicit affine framelet sets containing atoms implementable by sparse filters with directional characteristics. The algorithm below can easily be applied to every finite set of high-pass filters of our choice, multiplied by an appropriate set of scalars.

Specifically, for $n_{k}\in J\subset\mathbb{Z}^{s}$, let $H_{0}$ be a low-pass filter with positive coefficients $a=(a_{n_{k}})_{k=1}^{N}$ and $H$ be any high-pass filter of the form 
\[H(\cdot)=\begin{pmatrix}b_{n_{1}}^{1}&\ldots&b_{n_{N}}^{1}\\
\vdots&&\vdots\\
b_{n_{1}}^{L}&\cdots&b_{n_{N}}^{L}\end{pmatrix}\begin{pmatrix}e^{2\pi in_{1}\cdot}\\\vdots\\ e^{2\pi in_{N}\cdot}\end{pmatrix}\]
with $H(0)=0$.
\begin{description}
\item[Step 1:] We define the $1\times N$ vector $c=\left(\sqrt{a_{n_{k}}}\right)_{k=1}^{N}$ and notice that for any $\lambda\in\mathbb{R}^{L}$, the matrix
\[D_{1}(\lambda)=\diag(\lambda)\begin{pmatrix}b_{n_{1}}^{1}&\ldots&b_{n_{N}}^{1}\\
\vdots&&\vdots\\
b_{n_{1}}^{L}&\cdots&b_{n_{N}}^{L}\end{pmatrix}\begin{pmatrix}1/c_{n_{1}}&&\\&\ddots&\\&&1/c_{n_{N}}\end{pmatrix}\]
is well-defined and $D_{1}(\lambda)c^{T}=0$, since $H$ is a high-pass filter and therefore satisfies $\sum_{k=1}^{N}b_{n_{k}}^{i}=0$ for all $i=1,\ldots,L$.
\item[Step 2:] We use Theorem \ref{t3.2}(a) to obtain $\lambda^{*}$ such that
\[\tr\left(c^{T}c+D_{1}(\lambda^{*})^{T}D_{1}(\lambda^{*})\right)=\begin{cases}\max\tr\left(c^{T}c+D_{1}(\lambda)^{T}D_{1}(\lambda)\right) \\\text{subject to }\left\Vert c^{T}c+D_{1}(\lambda)^{T}D_{1}(\lambda) \right\Vert\leq 1\end{cases}\]
\item[Step 3:] We use Lemma \ref{l3.1} to find a completion matrix $D_{2}$ for which the rows of 
\[\begin{pmatrix}c\\D_{1}(\lambda^{*})\\D_{2}\end{pmatrix}\in\mathbb{R}^{(v+1)\times N},\quad v\geq N-1,\]
form a Parseval frame for $\mathbb{R}^{N}$.
\item[Step 4:] We use Theorem \ref{t2.5} to guarantee that the wavelets $\psi_{i}$ with corresponding high-pass filters $\lambda_{i}^{*}H_{i}$, $i=1,\ldots,L$ are components of a multi-wavelet $\Psi$ whose associated family $X_{\Psi}$ is a Parseval framelet for $L_{2}(\mathbb{R}^{s})$. Indeed, this follows from Theorem \ref{t2.5}(a), since the high-pass filter matrix $B$ is obtained by
\[B=\begin{pmatrix}D_{1}(\lambda^{*})\\D_{2}\end{pmatrix}\begin{pmatrix}c_{n_{1}}&&\\&\ddots&\\&&c_{n_{N}}\end{pmatrix}\]
\end{description} 

\begin{remark} \label{mainS4remark}
\begin{enumerate}
\item The cost of incorporating into $\Psi$ the frame wavelets defined by $\lambda_{i}^{*}H_{i}$ is paid in part by having to incorporate into $\Psi$ the filters that come from $D_{2}$. This cost can only be controlled if we select multiple high pass filters of our choice for which we have $rank(Q)=N$. This particular process will become more clear in what follows. 

\item The previous algorithm demonstrates the potentially limited role of the refinable function in the construction of $H_{1}$. Specifically, the algorithm shows that its main part can come from $H$. As we see, as long as $H$ has enough hand-picked filters to exhaust the available dimensionality of the construction space $\mathbb{R}^{N}$, the $D_{2}$-contribution in the high-pass filter set $H_{1}$ may be limited as measured by the reconstruction error $E(\cdot)$. Consequently, we are led to the conclusion that the significance of the refinable function is limited as the only role its seems to play is to set $N$.
\end{enumerate}
\end{remark}
In the spirit of the previous remark, we introduce the typical models of high-pass filter designs of our choice, including high-pass filters acting as \emph{first and second order directional finite-difference, Prewitt and Sobel operators}, known to produce desirable results in edge and singularity detection in 2-D imaging applications.

We recall that first and second order directional finite-difference filters are associated with the operators $\delta_{h,u}$ and $\delta_{h,u}^{2}$, respectively, where 
\[\delta_{h,u}[f](\cdot)=f(\cdot+hu)-f(\cdot-hu),\]
and
\[\delta_{h,u}^{2}[f](\cdot)=f(\cdot+hu)-2f(\cdot)+f(\cdot-hu).\]
In one dimension, the corresponding filter matrices are $(1,0,-1)$ and $(1,-2,1)$ (see \cite{Ron_1997_4616}). Those are used to generate tensor product filters, such as the Prewitt and Sobel filters \cite{HAST201465} given by
\[P_{x}=\begin{pmatrix}-1&0&1\\-1&0&1\\-1&0&1\end{pmatrix}\quad P_{y}=\begin{pmatrix}-1&-1&-1\\0&0&0\\1&1&1\end{pmatrix}\]
and
\[S_{x}=\begin{pmatrix}1&0&-1\\2&0&-2\\1&0&-1\end{pmatrix}\quad S_{y}=\begin{pmatrix}1&2&1\\0&0&0\\-1&-2&-1\end{pmatrix},\]
respectively. Both the Prewitt and Sobel operators are used to approximate or detect horizontal and vertical intensity changes. They are obtained as tensor products of smoothing and finite-difference operators, hence they are separable. We are interested in directing the action of such operators to several orientations to promote sparse decompositions and use them in feature extraction applications. For example, we notice that the matrices 
\begin{align*}
\begin{pmatrix}
0&0&0&1&0\\
0&0&0&0&0\\
0&0&0&0&0\\
0&0&0&0&0\\
0&-1&0&0&0
\end{pmatrix}
&\quad
\begin{pmatrix}
0&0&0&0&0\\
0&0&0&0&1\\
0&0&-2&0&0\\
1&0&0&0&0\\
0&0&0&0&0
\end{pmatrix}\\
\begin{pmatrix}
0&1&0&0&0\\
-1&0&1&0&0\\
0&-1&0&1&0\\
0&0&-1&0&1\\
0&0&0&-1&0
\end{pmatrix}
&\quad
\begin{pmatrix}
-1&0&1&0&0\\
0&0&0&0&0\\
0&-2&0&2&0\\
0&0&0&0&0\\
0&0&-1&0&1
\end{pmatrix}
\end{align*}
are sparse and oriented at $63.43^{\circ}$, $26.57^{\circ}$, $135^{\circ}$ and $116.57^{\circ}$, respectively, but cannot be obtained as tensor products of one-dimensional kernels. This is where our algorithm comes in handy, since it permits filters like the above to be part of filter families inducing Parseval framelets. Next, we construct families of wavelet frames arising from Cardinal $B$-spline refinable functions, whose low-pass filters have positive coefficients. 

For $N_{1}N_{2}=N$, let $h$ be an $N_{1}\times  N_{2}$ filter matrix. We define the map $\Lambda:\mathbb{R}^{N_{1}\times N_{2}}\to\mathbb{R}^{N}$ given by
\[\Lambda(h)=\left(h_{N_{1},1},\ldots,h_{N_{1},N_{2}},h_{N_{1}-1,1},\ldots,h_{N_{1}-1,N_{2}},\ldots,h_{1,1},\ldots,h_{1,N_{2}}\right)\in\mathbb{R}^{N}\]
to turn $h$ from a matrix to a vector, in accordance to Theorem \ref{t2.5}. As will become clear in examples \ref{ex2},\ref{ex3} and \ref{ex4}, we use $\Lambda$ in the following way: first, we pre-specify the form of a desirable high-pass filter matrix, say $h$, and then we define
\[d(\lambda):=\lambda \left(\frac{\Lambda(h)_{k}}{c_{n_{k}}}\right)_{k=1}^{N}\]
for a given vector $c=(c_{n_{k}})_{k=1}^{N}$. We then apply Steps 2,3 and 4 of our algorithm as stated above. When we do this for more than one filter $h$, then we must solve the optimization problem of Theorem \ref{t3.2}(a). If the filters we intend to use give pairwise orthogonal vectors through $\Lambda$, then the steps of the algorithm presented above can be applied to each filter individually.

The first case we examine is a high-pass filter family arising when we only apply Lemma \ref{l3.1} and Theorem \ref{t2.5}. In other words, we do not pre-design any of the filters.

\begin{example}
\label{ex1}
Let $\varphi$ be the one-dimensional second order cardinal $B$-spline refinable function with corresponding low-pass filter
\[\mu_{0}(\gamma)=\left(\frac{1+e^{2\pi i\gamma}}{2}\right)^{2}=\frac{1}{4}\left(1+2e^{2\pi i\gamma}+e^{4\pi i\gamma}\right), \quad \gamma\in\mathbb{T}\]
and consider $\phi$ to be the tensor product refinable function $\varphi\otimes\varphi$. Then $H_{0}(\gamma)=\mu_{0}(\gamma_{1})\mu_{0}(\gamma_{2})$ for $\gamma=(\gamma_{1},\gamma_{2})\in\mathbb{T}^{2}$ and the low-pass filter matrix is given by
\[h_{0}=\frac{1}{16}\begin{pmatrix}1&2&1\\2&4&2\\1&2&1\end{pmatrix}.\]
Using $\Lambda$, we define
\[c=\frac{1}{4}\left(1,\sqrt{2},1,\sqrt{2},2,\sqrt{2},1,\sqrt{2},1\right).\]
For symmetry purposes we translate $\phi$ so as to obtain $J=\{-1,0,1\}\times\{-1,0,1\}$. If we merely apply the SVD method of Lemma \ref{l3.1} we obtain
\begingroup\makeatletter\def\f@size{10}\check@mathfonts
\[B=10^{-2}\begin{pmatrix}
-8.84&31.8&-1.77&-3.54&-7.07&-3.54&-1.77&-3.54&-1.77\\ 
-6.25&-2.5&23.8&-2.5&-5&-2.5&-1.25&-2.5&-1.25\\ 
-8.84&-3.54&-1.77&31.8&-7.07&-3.54&-1.77&-3.54&-1.77\\ 
-12.5&-5&-2.5&-5&40&-5&-2.5&-5&-2.5\\ 
-8.84&-3.54&-1.77&-3.54&-7.07&31.8&-1.77&-3.54&-1.77\\ 
-6.25&-2.5&-1.25&-2.5&-5&-2.5&23.8&-2.5&-1.25\\ 
-8.84&-3.54&-1.77&-3.54&-7.07&-3.54&-1.77&31.8&-1.77\\ 
-6.25&-2.5&-1.25&-2.5&-5&-2.5&-1.25&-2.5&23.8\\ 
\end{pmatrix}\]
\endgroup
We notice that the fifth column of $B$ contains the constant terms in the generated high-pass filter polynomials. Based on this observation, we note that even though Theorem \ref{t2.5} guarantees that $B$ induces a Parseval frame for $L_{2}(\mathbb{R}^{2})$, none of the high-pass filter matrices are sparse, symmetric, anti-symmetric, or directional.
\end{example}

SVD for the construction of the high-pass filter set was first used in \cite{GOH2008} for proving the existence of periodic tight frame multiwavelets $L_2 ([0,2\pi)^s)$ arising from multi-refinable periodic functions. As we see, apart from generating compactly supported frame wavelets, there is essentially no luck in obtaining filters with some of the desirable properties by using SVD only.
 
\begin{example}
\label{ex2}
Let $\varphi$ be an even-order cardinal $B$-spline refinable function and let $\phi$ be the tensor product $\varphi\otimes\varphi$ as before, centered at the origin. Using $\Lambda$ and the fact that the symmetry of $h_{0}$ implies $a_{n_{i}}=a_{n_{N-i+1}}$ for $i=1,\ldots,(N-1)/2$, we define
\[Q=\begin{pmatrix}c\\ D_{1}\end{pmatrix}=\begin{pmatrix}
\sqrt{a_{n_{1}}}&\cdots&\sqrt{a_{n_{(N-1)/2}}}&\sqrt{a_{n_{(N+1)/2}}}&\sqrt{a_{n_{(N-1)/2}}}&\cdots&\sqrt{a_{n_{1}}}\\
-\frac{\sqrt{2}}{2}&\cdots&0&0&0&\cdots&\frac{\sqrt{2}}{2}\\
\vdots&&\vdots&\vdots&\vdots&&\vdots\\
0&\cdots&-\frac{\sqrt{2}}{2}&0&\frac{\sqrt{2}}{2}&\cdots&0
\end{pmatrix}.\]
We notice that $D_{1}$ defines central-difference filters with orientations parallel to the vectors $n_{i}$, $i=1,\ldots (N-1)/2$.
If $\beta$ is an arbitrary unit vector in $\mathbb{R}^{2}$, then we write 
\[cZ=\left((n_{1}\cdot\beta)\sqrt{a_{n_{1}}},\ldots,(n_{(N+1)/2}\cdot\beta)\sqrt{a_{n_{(N+1)/2}}},\ldots,(n_{N}\cdot\beta)\sqrt{a_{n_{1}}}\right)\]
as in Proposition \ref{p3.3} and note that the symmetry of the vectors $n_{i}$ and $n_{N-i+1}$ about the origin implies
\[n_{i}\cdot\beta=-n_{N-i+1}\cdot\beta,\quad i=1,\ldots,\frac{N-1}{2}.\]
This means that if a vector belongs to the orthogonal complement of the linear span of the rows of $Q$, then it is automatically orthogonal to $cZ$. In this setting, the rows of $Q$ are pairwise orthogonal unit vectors. Any choice of a $D_{2}$ matrix for which the rows of \[\begin{pmatrix}Q\\D_{2}\end{pmatrix}\]
form a Parseval frame for $\mathbb{R}^{N}$ will define an affine Parseval framelet for $L_{2}(\mathbb{R}^{2})$, where the $\psi_{i}$ defined by the rows of $D_{2}$ have exactly one directional vanishing moment for all $\beta\in\mathbb{R}^{2}$.

By Proposition \ref{p3.3}, each of the high-pass filters generated by $Q$ makes its corresponding wavelet insensitive to singularities parallel to $\beta$ when $\beta$ is perpendicular to $n_{k}$, since then the wavelet has infinite moments along these directions. In fact, by continuity of the inner product, each wavelet loses its sensitivity as $\beta$ converges to the unit vector perpendicular to $n_{k}$.
\end{example}

\begin{example}
\label{ex3}
Starting with the same refinable function $\phi$ as in example \ref{ex1}, our next effort is to design $B$ so that it is associated with four first-order and four second-order directional finite-difference high-pass filter matrices. Specifically, we consider the matrices
\begingroup\makeatletter\def\f@size{10}\check@mathfonts
\begin{align*}
&h_{1}=\begin{pmatrix}0&0&1\\0&0&0\\-1&0&0\end{pmatrix}, h_{2}=\begin{pmatrix}0&1&0\\0&0&0\\0&-1&0\end{pmatrix}, h_{3}=\begin{pmatrix}1&0&0\\0&0&0\\0&0&-1\end{pmatrix}, h_{4}=\begin{pmatrix}0&0&0\\-1&0&1\\0&0&0\end{pmatrix}\\
&h_{5}=\begin{pmatrix}0&0&1\\0&-2&0\\1&0&0\end{pmatrix}, h_{6}=\begin{pmatrix}0&1&0\\0&-2&0\\0&1&0\end{pmatrix}, h_{7}=\begin{pmatrix}1&0&0\\0&-2&0\\0&0&1\end{pmatrix}, h_{8}=\begin{pmatrix}0&0&0\\1&-2&1\\0&0&0\end{pmatrix},
\end{align*} 
\endgroup
which we vectorize using the map $\Lambda$ to obtain the rows of $D_{1}(\lambda)$ given by $d_{k}(\lambda)$, $k=1,\ldots,8$. This gives the matrix
\[D_{1}(\lambda):=\diag(\lambda)\begin{pmatrix}
-4&0&0&0&0&0&0&0&4\\
0&-2\sqrt{2}&0&0&0&0&0&2\sqrt{2}&0\\
0&0&-4&0&0&0&4&0&0\\
0&0&0&-2\sqrt{2}&0&2\sqrt{2}&0&0&0\\
0&0&0&-2\sqrt{2}&4&-2\sqrt{2}&0&0&0\\
0&0&-4&0&4&0&-4&0&0\\
0&-2\sqrt{2}&0&0&4&0&0&-2\sqrt{2}&0\\
-4&0&0&0&4&0&0&0&-4
\end{pmatrix}\]
whose rows are in the orthogonal complement of $c$. Here the rows of $D_{1}(\cdot)$ are not pairwise orthogonal and so the largest singular value of
\[Q(\lambda)=\begin{pmatrix}c\\D_{1}(\lambda)\end{pmatrix}\]
is expected to be strictly greater than $1$, even in the case where the rows of $Q$ are normalized. At this point, we invoke Theorem \ref{t3.2}(a). Specifically, we can find an optimal $\lambda^{*}$ so that $D_{1}(\lambda^{*})$ is a solution to
\[\begin{cases}\max\tr\left(c^{T}c+D_{1}^{T}(\lambda)D_{1}(\lambda)\right) \\ \text{subject to }\left\Vert c^{T}c+D_{1}^{T}(\lambda)D_{1}(\lambda)\right\Vert\leq 1\end{cases}.\]
We use Matlab's built-in function \emph{fmincon} to solve this problem and obtain
\[\lambda^{*}=\left(0.0442,0.0884,0.0442,0.0884,0.0234,0.0293,0.0088,0.0316\right),\]
but also the high-pass filter coefficients
\begingroup\makeatletter\def\f@size{10}\check@mathfonts
\[B=10^{-2}\begin{pmatrix}-17.7&0&0&0&0&0&0&0&17.7\\ 
0&-25&0&0&0&0&0&25&0\\ 
0&0&-17.7&0&0&0&17.7&0&0\\ 
0&0&0&-25&0&25&0&0&0\\ 
0&0&0&-6.63&13.26&-6.63&0&0&0\\ 
0&0&-11.75&0&23.5&0&-11.75&0&0\\ 
0&-2.5&0&0&5&0&0&-2.5&0\\ 
-12.65&0&0&0&25.3&0&0&0&-12.65\\ 
0.002&0&0.001&0.0003&-0.008&0.0003&0.001&0&0.002\\ 
-8.52&0.0288&9.59&0.233&-2.66&0.233&9.59&0.0288&-8.52\\ 
5.46&-0.939&5.69&-19&17.5&-19&5.69&-0.939&5.46\\ 
3.39&-21.5&3.4&8.1&13.2&8.1&3.4&-21.5&3.39\\ 
\end{pmatrix}\]
\endgroup
by Lemma \ref{l3.1} and Theorem \ref{t2.5}. The SVD process of Lemma \ref{l3.1} introduces four new filters, from the lower four rows of $B$, in order to complete the Parseval frame for $\mathbb{R}^{9}$. Moreover, as shown in example \ref{ex2}, the wavelets induced by the rows $\{b_{i}\}_{i=5}^{13}$ have first-order directional vanishing moments in the direction of all $\beta\in\mathbb{R}^{2}$. If we decide to omit the four filters added by $D_{2}$, Theorem \ref{t3.2}(b) implies that for an arbitrary function $f\in L_{2}(\mathbb{R}^{2})$, we have 
\[E(f)\leq(1-\sigma_{9}^{2})\Vert f\Vert_{L_{2}}^{2}\approx 0.987\Vert f\Vert_{L_{2}}^{2}.\]
Additionally, by Theorem \ref{t3.2}(b), the family
\[\{D_{2}^{j}T_{k}\psi_{i}: j\in\mathbb{Z},k\in\mathbb{Z}^{s},i=1,\ldots,8\}\]
is a frame, which guarantees the representation's injectivity. We also point out that, if all the row-vectors of $D_{1}(\lambda)$ are pairwise orthogonal, then the optimal $\lambda^{*}$ gives $\sigma_{i}(D_{1}(\lambda^{*}))=1$ for all $i$. The reader may refer to \cite{karSpie} for a Parseval framelet induced by the first five rows of $D_{1}(\lambda)$. In that paper we also present an application of the high-pass filter matrices arising from rows $3,4,5$ and $6$ of $B$ given by
\begin{align*}
&h_{3}=10^{-2}\begin{pmatrix}17.7&0&0\\ 
0&0&0\\ 
0&0&-17.7\\ 
\end{pmatrix}&h_{4}=10^{-2}\begin{pmatrix}0&0&0\\ 
-25&0&25\\ 
0&0&0\\ 
\end{pmatrix}\\
&h_{4}=10^{-2}\begin{pmatrix}0&0&0\\ 
-6.63&13.26&-6.63\\ 
0&0&0\\ 
\end{pmatrix}&h_{6}=10^{-2}\begin{pmatrix}-11.75&0&0\\ 
0&23.5&0\\ 
0&0&-11.75\\ 
\end{pmatrix}
\end{align*}
\begin{figure}[H] 	
\centering
\includegraphics[width=0.27\textwidth,height=0.27\textwidth]{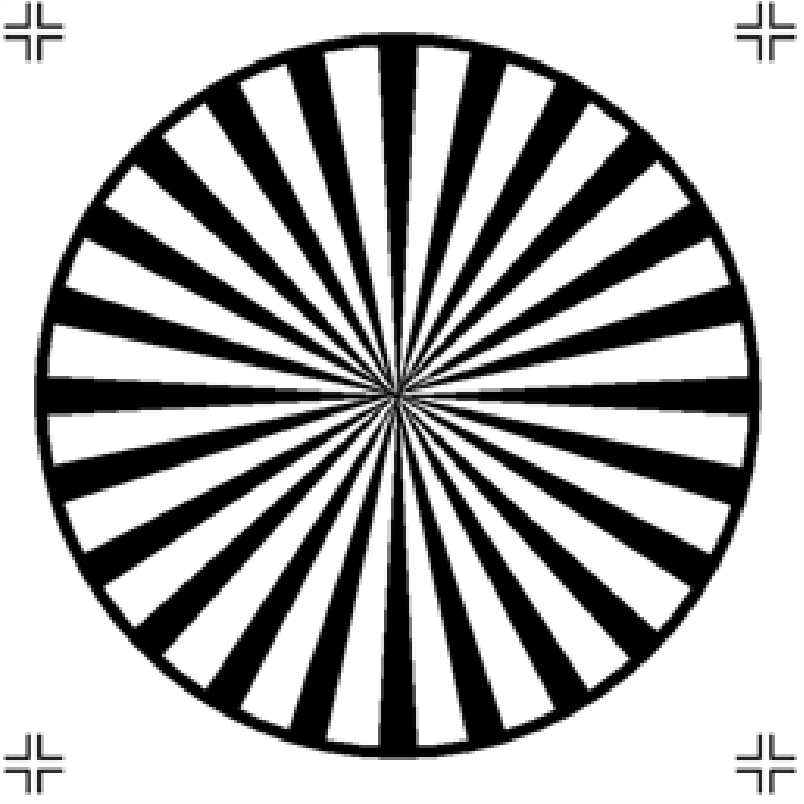}
\caption{ This is a 256x256 image freely available with Matlab 2017. We use it to demonstrate the interaction of the designed filters with singularities in various directions.}
\end{figure}

\begin{figure}[H]
\centering
\includegraphics[width=0.3\textwidth,height=0.23\textwidth]{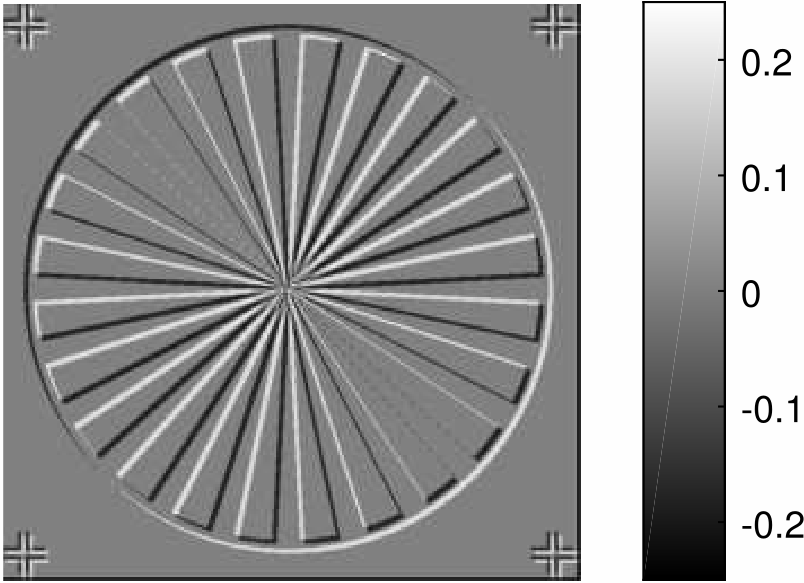}\qquad\qquad\qquad
\includegraphics[width=0.3\textwidth,height=0.23\textwidth]{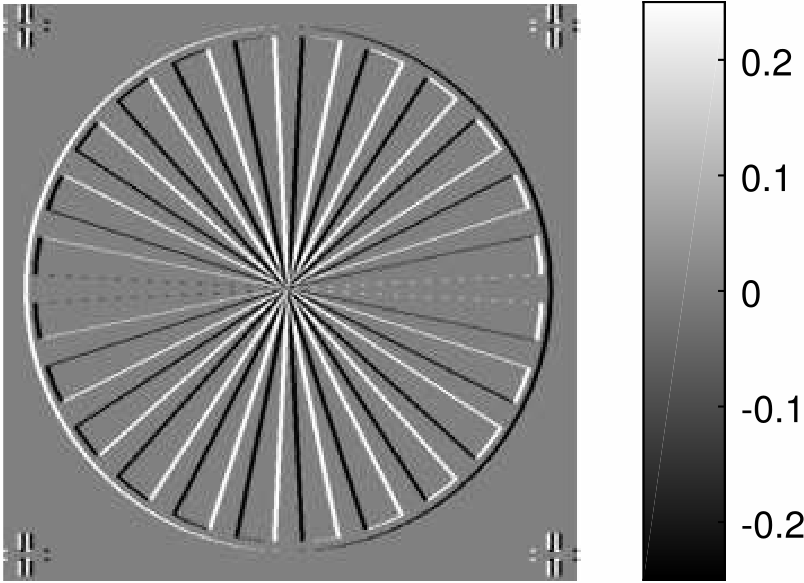}\end{figure}
\medskip
\medskip
\begin{figure}[H]
\centering
\includegraphics[width=0.3\textwidth,height=0.23\textwidth]{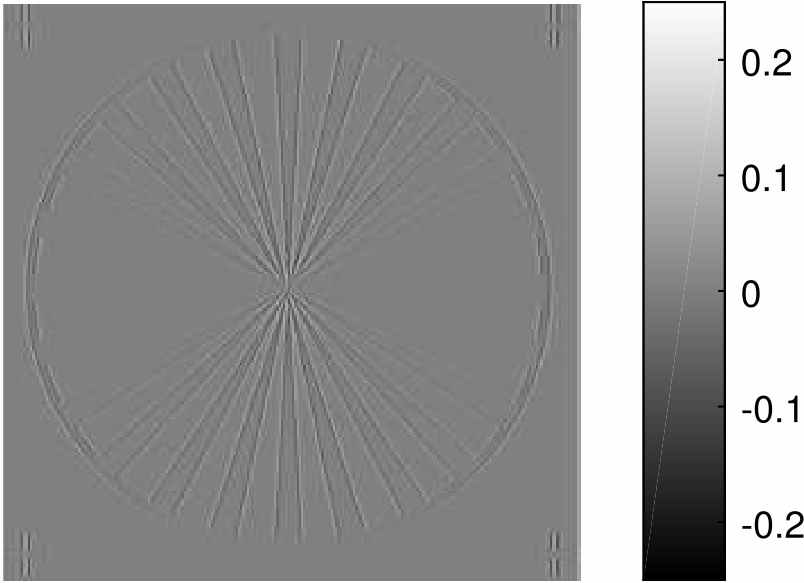}\qquad\qquad\qquad
\includegraphics[width=0.3\textwidth,height=0.23\textwidth]{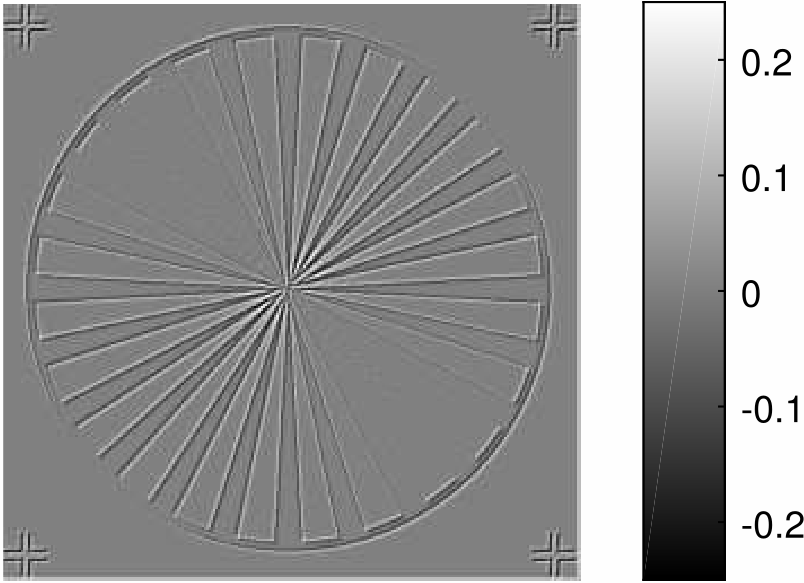}
\caption{Application of $h_{i}$, $i=3,4,5,6$ constructed in Example \ref{ex3} as discrete 2D-convolution kernels at native reolution. The first two filters act as first-order directional central-difference filters oriented at $135^{\circ}$ and $0^{\circ}$, respectively. The last two act as second-order central-difference filters oriented at $0^{\circ}$ and $135^{\circ}$, respectively. Note that singularity detection strength increases as edges are oriented closer to being perpendicular to the orientation of each filter. In Fig. \ref{fig5} we see that this effect may also be related to the anisotropy of the filter and its size.}
\end{figure}
\end{example}

\begin{example}
\label{ex4}
We consider the fourth order cardinal $B$-spline refinable function 
\[\varphi(x)=\begin{cases}\frac{1}{6}x^{3},& 0<x\leq 1\\
\frac{1}{6}(-x^{3}+12x^{2}-18x+8),& 1<x\leq 2\\
\frac{1}{6}(-x^{3}-12x^{2}+78x-88),& 2<x\leq 3\\
\frac{1}{6}(x^{3}-48x+128),& 3<x\leq 4\end{cases}\]
with corresponding low-pass filter
\[\mu_{0}(\gamma)=\left(\frac{1+e^{2\pi i\gamma}}{2}\right)^{4}=\frac{1}{16}\left(1+4e^{2\pi i\gamma}+6e^{4\pi i\gamma}+4e^{6\pi i\gamma}+e^{8\pi i\gamma}\right),\]
and we set $\phi$ to be the tensor product $\varphi\otimes\varphi$. Then $H_{0}(\gamma)=\mu_{0}(\gamma_{1})\mu_{0}(\gamma_{2})$, the low-pass filter matrix is given by
\[h_{0}=\frac{1}{64}\begin{pmatrix}
1&4&6&4&1\\4&16&24&16&4\\6&24&36&24&6\\4&16&24&16&4\\1&4&6&4&1
\end{pmatrix}\]
and $c$ takes the form
\[c=\frac{1}{16}\left(1,2,\sqrt{6},2,1,2,4,2\sqrt{6},4,2,\sqrt{6},2\sqrt{6},6,2\sqrt{6},\sqrt{6},2,4,2\sqrt{6},4,2,1,2,\sqrt{6},2,1\right).\]
Centering $\phi$ at the origin implies $J=\{-2,\ldots,2\}\times\{-2,\ldots,2\}$. We use our algorithm to create filters with different orientations from those along which their corresponding finite-difference kernels act. More specifically, we consider first and second-order filters of the form
\begin{align*}
&\begin{pmatrix}0&0&0&-1&0\\ 
0&0&-1&0&1\\ 
0&-1&0&1&0\\ 
-1&0&1&0&0\\ 
0&1&0&0&0\\ 
\end{pmatrix}&\begin{pmatrix}0&0&-1&0&1\\ 
0&0&0&0&0\\ 
0&-1&0&1&0\\ 
0&0&0&0&0\\ 
-1&0&1&0&0\\ 
\end{pmatrix}\\
&\begin{pmatrix}0&-1&0&1&0\\ 
0&-1&0&1&0\\ 
0&-1&0&1&0\\ 
0&-1&0&1&0\\ 
0&-1&0&1&0\\ 
\end{pmatrix}&\begin{pmatrix}-1&0&1&0&0\\ 
0&0&0&0&0\\ 
0&-1&0&1&0\\ 
0&0&0&0&0\\ 
0&0&-1&0&1\\ 
\end{pmatrix}\\
&\begin{pmatrix}0&0&0&1&-1\\ 
0&0&1&-2&1\\ 
0&1&-2&1&0\\ 
1&-2&1&0&0\\ 
-1&1&0&0&0\\ 
\end{pmatrix}&\begin{pmatrix}0&0&1&-2&1\\ 
0&0&0&0&0\\ 
0&1&-2&1&0\\ 
0&0&0&0&0\\ 
1&-2&1&0&0\\ 
\end{pmatrix}\\
&\begin{pmatrix}0&1&-2&1&0\\ 
0&1&-2&1&0\\ 
0&1&-2&1&0\\ 
0&1&-2&1&0\\ 
0&1&-2&1&0\\ 
\end{pmatrix}&\begin{pmatrix}1&-2&1&0&0\\ 
0&0&0&0&0\\ 
0&1&-2&1&0\\ 
0&0&0&0&0\\ 
0&0&1&-2&1\\ 
\end{pmatrix}.
\end{align*}
First, with this new design approach we mimic one of the popular properties of curvelets and shearlets: We define filters that act as singularity detectors perpendicularly to the local orientation of a wavefront. Since our design is limited within $J$, the discreteness of this spatially limited integer subgrid constrains our ability to direct the action of the associated differential operator perpendicularly to the filter's orientation. Moreover, the smaller number of bands of the filter matrix relative to the length along its orientation seems to better focus the direction of its action (see Fig. \ref{fig5}). This is something we also observe to a greater degree with shearlets and curvelets, because they are designed in the frequency domain where one can control their shape more easily.

The prototype of each of the two classes of the filters we design in this example is directed along the $x$ or $y$ axis. The third and seventh matrices above are the prototype filters for the first and second order directional central difference operators acting along the $x$ direction. Both filters have vertical orientation. To switch these filters to another orientation, we reposition their 
 central band by selecting one-by-one the lead point of the central band on the $x$ and $y$-axis of the grid as shown in Figure \ref{reposition} below.

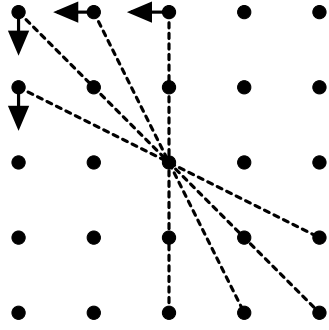
\begin{figure}[H]
\centering
\begin{tikzpicture}[line cap=round,line join=round,>=triangle 45,x=1.0cm,y=1.0cm]
\clip(-2.5,-2.5) rectangle (2.5,2.5);
\draw [line width=1.2pt,dash pattern=on 2pt off 2pt] (0.,2.)-- (0.,-2.);
\draw [->,line width=1.2pt] (0.,2.) -- (-0.56,2.);
\draw [line width=1.2pt,dash pattern=on 2pt off 2pt] (-1.,2.)-- (1.,-2.);
\draw [->,line width=1.2pt] (-1.,2.) -- (-1.54,2.);
\draw [line width=1.2pt,dash pattern=on 2pt off 2pt] (-2.,2.)-- (2.,-2.);
\draw [->,line width=1.2pt] (-2.,2.) -- (-2.,1.42);
\draw [line width=1.2pt,dash pattern=on 2pt off 2pt] (-2.,1.)-- (2.,-1.);
\draw [->,line width=1.2pt] (-2.,1.) -- (-2.,0.42);
\begin{scriptsize}
\draw [fill=black] (0.,1.) circle (2.5pt);
\draw [fill=black] (1.,1.) circle (2.5pt);
\draw [fill=black] (-2.,-1.) circle (2.5pt);
\draw [fill=black] (-2.,0.) circle (2.5pt);
\draw [fill=black] (-2.,-2.) circle (2.5pt);
\draw [fill=black] (1.,-1.) circle (2.5pt);
\draw [fill=black] (2.,-1.) circle (2.5pt);
\draw [fill=black] (0.,0.) circle (2.5pt);
\draw [fill=black] (1.,0.) circle (2.5pt);
\draw [fill=black] (1.,2.) circle (2.5pt);
\draw [fill=black] (0.,2.) circle (2.5pt);
\draw [fill=black] (-1.,1.) circle (2.5pt);
\draw [fill=black] (-1.,2.) circle (2.5pt);
\draw [fill=black] (-1.,0.) circle (2.5pt);
\draw [fill=black] (-2.,1.) circle (2.5pt);
\draw [fill=black] (0.,-2.) circle (2.5pt);
\draw [fill=black] (1.,-2.) circle (2.5pt);
\draw [fill=black] (2.,0.) circle (2.5pt);
\draw [fill=black] (2.,1.) circle (2.5pt);
\draw [fill=black] (2.,2.) circle (2.5pt);
\draw [fill=black] (-2.,2.) circle (2.5pt);
\draw [fill=black] (-1.,-2.) circle (2.5pt);
\draw [fill=black] (2.,-2.) circle (2.5pt);
\draw [fill=black] (-1.,-1.) circle (2.5pt);
\draw [fill=black] (0.,-1.) circle (2.5pt);
\end{scriptsize}
\end{tikzpicture}
\caption{The dashed lines show four successive positions of central bands defining this pre-designed filter set. Once the central band has been set, we choose its nearest diametrically opposite bands to create all first and second-order finite difference filters allowed by this process.}
\label{reposition}
\end{figure} 

This process gives a filter bank with $24$ high pass filters with hand-picked orientations. Next, SVD adds $24$ more to complete a Parseval frame. The full list of all 48 filters of this example and of Example \ref{ex3} can be found in the supplementary file which can be retrieved from \url{https://github.com/nkarantzas/multi-d-compactly-supported-PF-} along with the codes used for the generation of the presented filter-banks.
\medskip
\medskip
\begin{figure}[H]
\centering
\includegraphics[width=0.3\textwidth,height=0.22\textwidth]{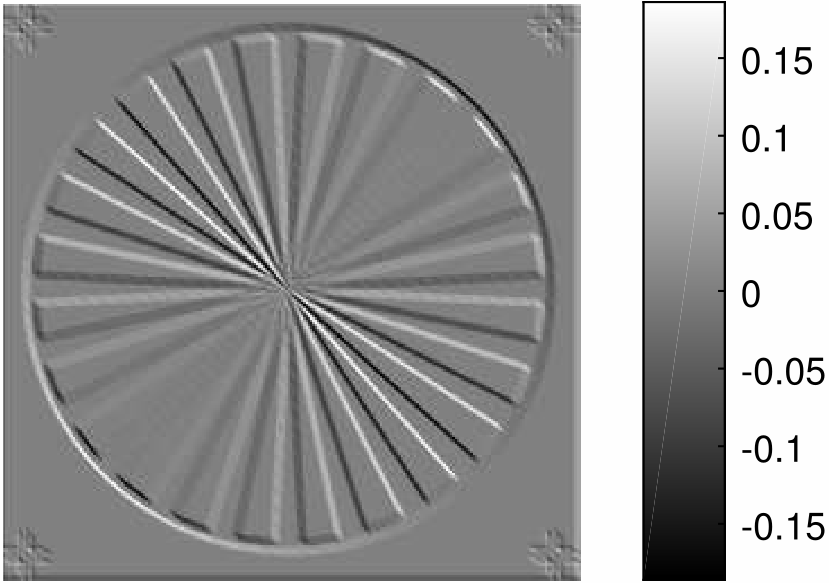}\qquad\qquad\qquad
\includegraphics[width=0.3\textwidth,height=0.22\textwidth]{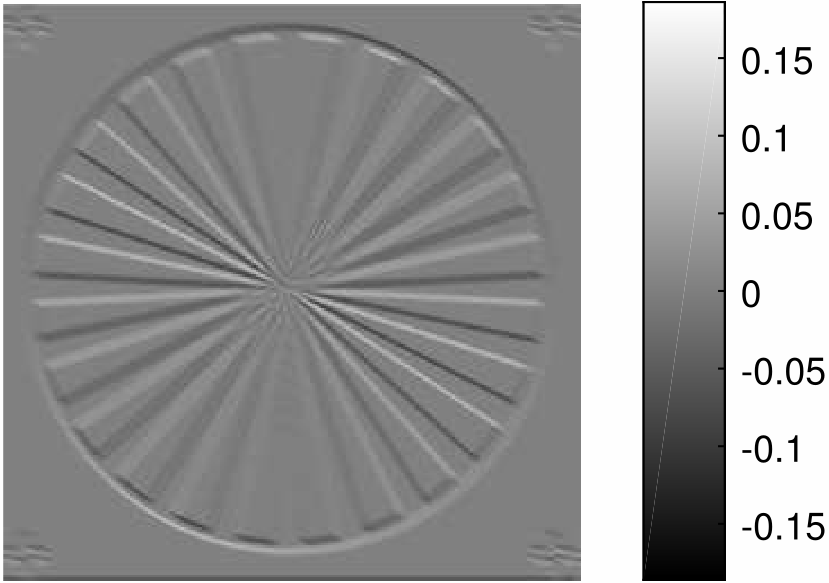}
\end{figure}
\medskip
\medskip
\begin{figure}[H]
\centering
\includegraphics[width=0.3\textwidth,height=0.22\textwidth]{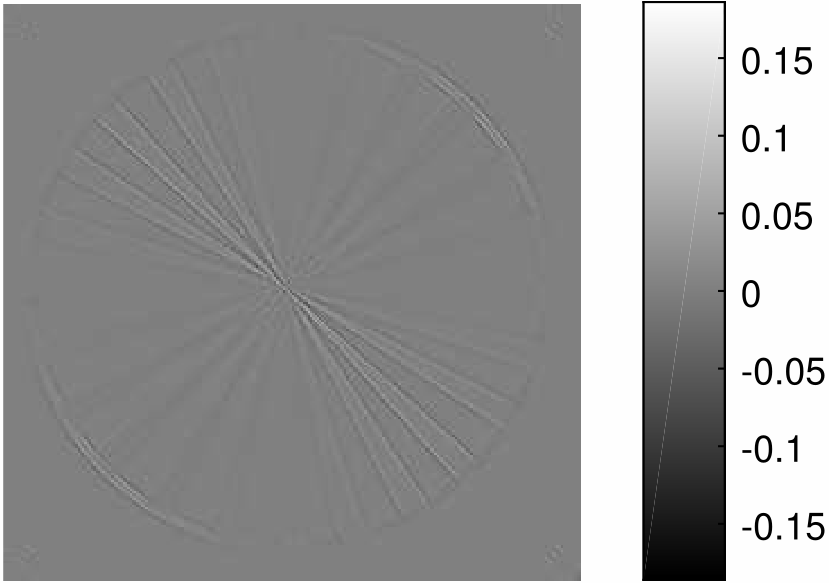}\qquad\qquad\qquad
\includegraphics[width=0.3\textwidth,height=0.22\textwidth]{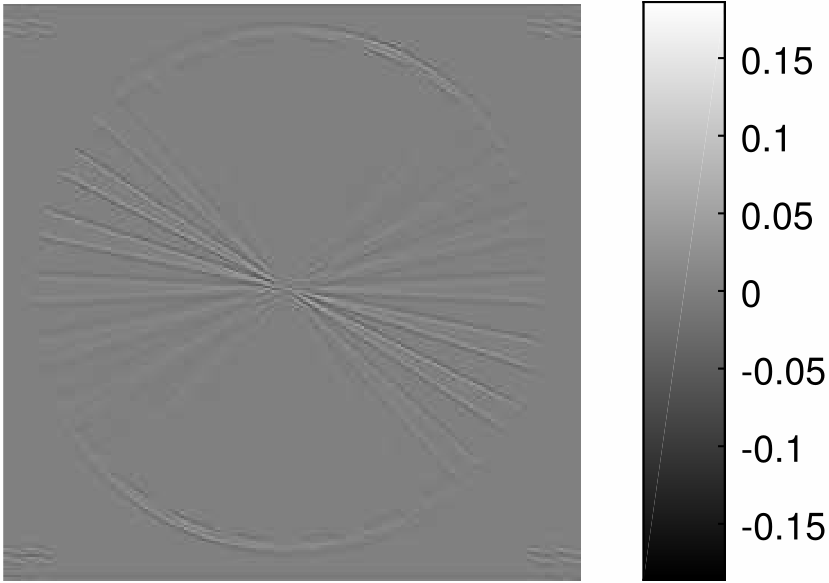}
\caption{Application of $h_{i}$, $i=5,6,17,18$ at native resolution. The first two convolutions correspond to filters with orientations at $135^{\circ}$ and $153.43^{\circ}$, respectively. The last two convolutions correspond to filters with orientations at $135^{\circ}$ and $153.43^{\circ}$, respectively.}
\label{fig5}
\end{figure}
\end{example}

\begin{example}
\label{ex5} 
As promised in Section \ref{s2}, we illustrate the geometric implications and complexities of solving the system of equations \eqref{4} and \eqref{5}. Equation \eqref{5} is relevant only when $M$ is not a diagonal matrix. Recall that our analysis in Sections \ref{s2} and \ref{s3} is based on $M$ being diagonal. To avoid computational complications, we consider the one-dimensional case, i.e., $s=1$. Without loss of generality, we assume $\{n_{k}\}_{k=1}^{N}$ are consecutive integers. Then
\[\begin{cases}
e^{2\pi i(n_{2}-n_{1})\gamma}=e^{2\pi i(n_{3}-n_{2})\gamma}=\ldots=e^{2\pi i(n_{N}-n_{N-1})\gamma}\\
e^{2\pi i(n_{3}-n_{1})\gamma}=e^{2\pi i(n_{4}-n_{2})\gamma}=\ldots=e^{2\pi i(n_{N}-n_{N-2})\gamma}\\
\vdots\\
e^{2\pi i(n_{N-1}-n_{1})\gamma}=e^{2\pi i(n_{N}-n_{2})\gamma}.
\end{cases}\]
The above equalities indicate that by rearranging and regrouping the monomials in \eqref{5} with respect to a fixed-valued $n_{t}-n_{k}$, we conclude that equation \eqref{5} is satisfied if and only if 
\[\sum_{k=1}^{N-t}m_{k,k+t}e^{-2\pi in_{k+t}q}=0,\]
for all $t=1,\ldots,N-1$, which along with equation \eqref{4} give a full characterization of the problem.

However, even though the above equation indicates there is a relationship between the elements of the $j$-th off-diagonal of the matrix $M$, it does not provide us with any insight on the dimension of the desired high-pass vector, or a definite way of acquiring it.

For example, in the setting of the classical construction of orthonormal wavelets, let $H_{0}$ be a low-pass filter with $4$ coefficients given by $a=[a_{1},a_{2},a_{3},a_{4}]$ and $H_{1}$ be a high-pass filter with coefficients $B=[b_{1},b_{2},b_{3},b_{4}]$. Since $M=a^{T}a+B^{T}B$ is symmetric, the previous system of equations is equivalent to 
\[\begin{cases}
m_{11}+m_{22}+m_{33}+m_{44}=1,\\
m_{11}-m_{22}+m_{33}-m_{44}=0,\\
m_{12}+m_{23}+m_{34}=0,\\
m_{12}-m_{23}+m_{34}=0,\\
m_{13}+m_{24}=0,\\
m_{13}-m_{24}=0,\\
m_{14}=0,
\end{cases}\]
from which we deduce $m_{13}=m_{14}=m_{23}=m_{24}=0$ and $m_{12}=-m_{34}$. Now let $v_{k}\in\mathbb{R}^{2}$, $k=1,2,3,4$ be the column vectors of
\[\begin{pmatrix}a\\B\end{pmatrix}=\begin{pmatrix}a_{1}&a_{2}&a_{3}&a_{4}\\b_{1}&b_{2}&b_{3}&b_{4}\end{pmatrix}.\]
Then the above linear system suggests
\begin{itemize}
\item $v_{1}$ is orthogonal to $v_{3}$ and $v_{4}$, and $v_{2}$ is orthogonal to $v_{3}$ and $v_{4}$. Hence $m_{12}\neq 0$, $v_{1}\parallel v_{2}$ and $v_{3}\parallel v_{4}$.
\item Finally, since $m_{12}=-m_{34}$, if $v_{1}$ and $v_{2}$ are parallel, $v_{3}$ and $v_{4}$ must be anti-parallel and vice versa.
\end{itemize}
This analysis indicates that the vectors $v_{k}$ can only form a capital T-shaped configuration as indeed they do, for example in the Daubechies $D4$ case \cite{Daub1} where the corresponding matrix $\begin{pmatrix}a\\B\end{pmatrix}$ is given by 
\[\begin{pmatrix}a\\B\end{pmatrix}=\frac{1}{8}\begin{pmatrix}1+\sqrt{3}&3+\sqrt{3}&3-\sqrt{3}&1-\sqrt{3}\\1-\sqrt{3}&\sqrt{3}-3&3+\sqrt{3}&-1-\sqrt{3}\end{pmatrix}.\] 
\begin{figure}[H]
\centering
\begin{tikzpicture}[scale=4.0]
\draw[->,color=black] (-0.2996,0.) -- (1.,0.);
\foreach \x in {-0.2,0.2,0.4,0.6,0.8}
\draw[shift={(\x,0)},color=black] (0pt,-2pt);
\draw[color=black] (0.9521574193548397,0.0116914485165794) node [anchor=south west] { x};
\draw[->,color=black] (0.,-0.5) -- (0.,0.85);
\foreach \y in {-0.4,-0.2,0.2,0.4,0.6,0.8}
\draw[shift={(0,\y)},color=black] (-2pt,0pt);
\draw[color=black] (0.01504516129032259,0.8230970331588128) node [anchor=west] { y};
\clip(-0.2996,-0.5) rectangle (1.,0.85);
\draw [->,line width=1.2pt] (0.,0.) -- (0.4829629131445341,-0.12940952255126034);
\draw [->,line width=1.2pt] (0.,0.) -- (0.8365163037378077,-0.2241438680420134);
\draw [->,line width=1.2pt] (0.,0.) -- (0.2241438680420134,0.8365163037378077);
\draw [->,line width=1.2pt] (0.,0.) -- (-0.12940952255126034,-0.4829629131445341);
\draw (0.16,-0.1) node[anchor=north west] {$v_{1}$};
\draw (0.6,-0.23) node[anchor=north west] {$v_{2}$};
\draw (0.22,0.78) node[anchor=north west] {$v_{3}$};
\draw (-0.28,-0.35) node[anchor=north west] {$v_{4}$};
\end{tikzpicture}
\label{fig:D4}
\end{figure}
Finally, we notice that if one wants to have additional high-pass filters or increase the length of the filters, the number of degrees of freedom increases significantly and the problem of maintaining a geometric intuition of the underlying properties becomes more complex. Moreover, we note that in the case of a four non-zero coefficient low-pass filter, we cannot have only non-negative coefficients.
\end{example}
\section{Acknowledgment}  This work was partially supported by NSF with award NSF-DMS 1720487 and NSF-DMS 1320910.

\bibliographystyle{unsrt}
\bibliography{bibsync}

\begin{thebibliography}{10}

\bibitem{Daubechies_1992_4532}
I.~Daubechies.
\newblock {\em Ten lectures on wavelets}.
\newblock Number~61 in CBMS. SIAM: Society for Industrial and Applied
  Mathematics, 1992.

\bibitem{Vetterli_1995_4637}
M.~Vetterli and J.~Kovacevic.
\newblock {\em Wavelets and subband coding}.
\newblock Prentice Hall PTR, Englewood Cliffs, NJ, 1995.

\bibitem{Kovacevic_1992_4518}
J.~Kovacevic and M.~Vetterli.
\newblock Nonseparable multidimensional perfect reconstruction filter-banks.
\newblock {\em IEEE Trans. Inf. Theory}, 38:533--555, 1992.

\bibitem{Ayache_2001_6383}
A.~Ayache.
\newblock Some methods for constructing nonseparable, orthonormal, compactly
  supported wavelet bases.
\newblock {\em Applied and Computational Harmonic Analysis}, 10:99--111, 2001.

\bibitem{Belogay_1999_58}
Eugene Belogay and Yang Wang.
\newblock Arbitrarily smooth orthogonal nonseparable wavelets in
  $\mathbb{R}^{2}$.
\newblock {\em SIAM J. Math. Anal.}, 30(3):678--697, 1999.

\bibitem{Candes2006}
Emmanuel Candes, Laurent Demanet, David Donoho, and Lexing Ying.
\newblock Fast discrete curvelet transforms.
\newblock {\em Multiscale Model. Simul.}, 5(3):861--899, 2006.

\bibitem{Candes2005}
Emmanuel~J Candes and Laurent Demanet.
\newblock The curvelet representation of wave propagators is optimally sparse.
\newblock {\em Comm. Pure Appl. Math.}, 58(11):1472--1528, 2005.

\bibitem{Demanet_2003_4534}
L.~Demanet and P.~Vandergheynst.
\newblock Gabor wavelets on the sphere.
\newblock In {\em Proc. SPIE Int. Soc. Opt. Eng.}, volume 5207, pages 5207 --
  5207 -- 8, 2003.

\bibitem{Candes_1999_6387}
E.~J. Candes.
\newblock Harmonic analysis of neural netwoks.
\newblock {\em Appl. Comput. Harmon. Anal}, 6:197--218, 1999.

\bibitem{Candes_2004_12186}
Emmanuel~J. Cand{\`e}s and David~L. Donoho.
\newblock New tight frames of curvelets and optimal representations of objects
  with piecewise {$C\sp 2$} singularities.
\newblock {\em Comm. Pure Appl. Math.}, 57(2):219--266, 2004.

\bibitem{Labate2005}
D.~Labate, W.~Lim, G.~Kutyniok, and G.~Weiss.
\newblock Sparse multidimensional representation using shearlets.
\newblock {\em SPIE Proc. 5914, SPIE, Bellingham}, pages 254--262, 2005.

\bibitem{Ron_1997_4616}
A.~Ron and Z.~Shen.
\newblock Affine system in $\mathcal{L}^{2}(\mathbb{R}^{d})$: {T}he analysis of
  the analysis operator.
\newblock {\em J. Funct. Anal.}, (148):408--447, 1997.

\bibitem{Ron_1997_4615}
A.~Ron and Z.~Shen.
\newblock Affine systems in ${{L_2}({\mathbb R}^d)}$ {I}{I}: {D}ual systems.
\newblock {\em J. Fourier Anal. Appl.}, 3:617--637, 1997.

\bibitem{Grohs2014}
Philipp Grohs and Gitta Kutyniok.
\newblock Parabolic molecules.
\newblock {\em Found. Comput. Math.}, 14(2):299--337, Apr 2014.

\bibitem{DAUBECHIES20031}
Ingrid Daubechies, Bin Han, Amos Ron, and Zuowei Shen.
\newblock Framelets: Mra-based constructions of wavelet frames.
\newblock {\em Appl. Comput. Harmon. Anal.}, 14(1):1 -- 46, 2003.

\bibitem{Han2012}
Bin Han.
\newblock Nonhomogeneous wavelet systems in high dimensions.
\newblock {\em Appl. Comput. Harmon. Anal.}, 32(2):169--196, 2012.

\bibitem{CHUI2002224}
Charles~K. Chui, Wenjie He, and Joachim St{\"o}ckler.
\newblock Compactly supported tight and sibling frames with maximum vanishing
  moments.
\newblock {\em Appl. Comput. Harmon. Anal.}, 13(3):224 -- 262, 2002.

\bibitem{ATREAS201451}
Nikolaos Atreas, Antonios Melas, and Theodoros Stavropoulos.
\newblock Affine dual frames and extension principles.
\newblock {\em Appl. Comput. Harmon. Anal.}, 36(1):51 -- 62, 2014.

\bibitem{Atreas2016}
Nikolaos~D. Atreas, Manos Papadakis, and Theodoros Stavropoulos.
\newblock Extension principles for dual multiwavelet frames of
  $\mathcal{L}^{2}(\mathbb{R}^{s})$ constructed from multirefinable generators.
\newblock {\em J. Fourier Anal. Appl.}, pages 1--24, 2016.

\bibitem{Kutyniok_2009_12187}
Gitta Kutyniok and Demetrio Labate.
\newblock Resolution of the wavefront set using continuous shearlets.
\newblock {\em Trans. Amer. Math. Soc.}, 361(5):2719--2754, 2009.

\bibitem{Kittipoom2012}
Pisamai Kittipoom, Gitta Kutyniok, and Wang-Q. Lim.
\newblock Construction of compactly supported shearlet frames.
\newblock {\em Constr. Approx.}, 35(1):21--72, Feb 2012.

\bibitem{Ayache_1997_29}
A.~Ayache.
\newblock Construction de bases orthonorm\'es d'ondelettes de
  $\mathbb{L}^{2}(\mathbb{R}^{2})$ non s\'eparables, à support compact et de
  r\'egularit\'e arbitrairement grande.
\newblock {\em Comptes Rendus Acad\'emie des Sciences de Paris}, 325:17--20,
  1997.

\bibitem{ehler05}
Martin Ehler.
\newblock Compactly supported multivariate wavelet frames obtained by
  convolution.
\newblock 2005.

\bibitem{SANANTOLIN2013201}
A.~San Antol\'in and R.A. Zalik.
\newblock A family of nonseparable scaling functions and compactly supported
  tight framelets.
\newblock {\em Journal of Mathematical Analysis and Applications}, 404(2):201
  -- 211, 2013.

\bibitem{HAN200343}
Bin Han.
\newblock Compactly supported tight wavelet frames and orthonormal wavelets of
  exponential decay with a general dilation matrix.
\newblock {\em Journal of Computational and Applied Mathematics}, 155(1):43 --
  67, 2003.
\newblock Approximation Theory, Wavelets, and Numerical Analysis.

\bibitem{Han_Jiang_Shen_Zhuang}
Bin Han, Qingtang Jiang, Zuowei Shen, and Xiaosheng Zhuang.
\newblock Symmetric canonical quincunx tight framelets with high vanishing
  moments and smoothness.
\newblock {\em Math. Comput.}, 87:347--379, 2018.

\bibitem{HAN1997380}
Bin Han.
\newblock On dual wavelet tight frames.
\newblock {\em Applied and Computational Harmonic Analysis}, 4(4):380 -- 413,
  1997.

\bibitem{kingsbury_1999_394}
N.~Kingsbury.
\newblock Image processing with complex wavelets.
\newblock {\em Phil. Trans. R. Soc. London A}, 357:2543--2560, 1999.

\bibitem{Selesnick_2000_4620}
I.W. Selesnick and L.~Sendur.
\newblock Iterated oversampled filter banks and wavelet frames.
\newblock In M.~Unser A.~Aldroubi, A.~Laine, editor, {\em Proc. Wavelet
  Applications in Signal and Image Processing VIII}, volume 4119 of {\em
  Proceedings of SPIE}, 2000.

\bibitem{Han_Zhao}
B.~Han and Z.~Zhao.
\newblock Tensor product complex tight framelets with increasing
  directionality.
\newblock {\em SIAM Journal on Imaging Sciences}, 7(2):997--1034, 2014.

\bibitem{Han_Mo_Zhao}
B.~Han, Q.~Mo, and Z.~Zhao.
\newblock Compactly supported tensor product complex tight framelets with
  directionality.
\newblock {\em SIAM Journal on Mathematical Analysis}, 47(3):2464--2494, 2015.

\bibitem{Cabrelli_2000_4507}
C.~A. Cabrelli and M-L. Gordillo.
\newblock Existence of multiwavelets in $\mathbb{R}^{n}$.
\newblock {\em Proc. Amer. Math. Soc.}, 130(5):1413--1424, 2000.

\bibitem{Adelson_1987_3}
E.~H. Adelson, E.~Simoncelli, and R.~Hingoranp.
\newblock Orthogonal pyramid transforms for image coding.
\newblock {\em Visual Communications and Image Processing II}, 845:50--58,
  1987.

\bibitem{Simoncelli_1995_606}
E.P. Simoncelli and W.T. Freeman.
\newblock The steerable pyramid: {A} flexible architecture for multi-scale
  derivative computation.
\newblock {\em Proc. IEEE International Conference on Image Processing}, 1995.

\bibitem{Candes_1999_4514}
E.J. Candes and D.L. Donoho.
\newblock Ridgelets: {A} key to higher dimensional intermittency?
\newblock {\em Phil. Trans. R. Soc. London}, A:2495--2509, 1999.

\bibitem{Guo_2007_12169}
Kanghui Guo and Demetrio Labate.
\newblock Optimally sparse multidimensional representation using shearlets.
\newblock {\em SIAM J. Math. Anal.}, 39:298--318, 2007.

\bibitem{Papadakis_2009_4602}
M.~Papadakis, B.G. Bodmann, S.K. Alexander, D.~Vela, S.~Baid, A.A. Gittens,
  D.J. Kouri, S.D. Gertz, S.~Jain, J.R. Romero, X.~Li, P.~Cherukuri, D.D. Cody,
  G.W. Gladish, Aboshady., J.L. Conyers, and S.W. Casscells.
\newblock Texture-based tissue characterization for high-resolution {CT}-scans
  of coronary arteries.
\newblock {\em Commun. Numer. Methods Eng.}, 25(6):597--613, 2009.

\bibitem{BinHan2017}
Bin Han, Tao Li, and Xiaosheng Zhuang.
\newblock Directional compactly supported box spline tight framelets with
  simple geometric structure.
\newblock {\em Applied Mathematics Letters}, 91:213 -- 219, 2019.

\bibitem{MinhDo05}
Yue Lu and M.~N. Do.
\newblock The finer directional wavelet transform.
\newblock In {\em Proc. IEEE Int. Conf. Acoust. Speech Signal Process.},
  volume~4, pages iv/573--iv/576, March 2005.

\bibitem{MinhDo3}
Y.~M. Lu and M.~N. Do.
\newblock Multidimensional directional filter banks and surfacelets.
\newblock {\em IEEE Trans. Image Process.}, 16(4):918--931, April 2007.

\bibitem{MinhDo07}
A.~L. da~Cunha and M.~N. Do.
\newblock On two-channel filter banks with directional vanishing moments.
\newblock {\em IEEE Trans. Image Process.}, 16(5):1207--1219, May 2007.

\bibitem{MinhDo4}
A.~L. da~Cunha and M.~N. Do.
\newblock Bi-orthogonal filter banks with directional vanishing moments.
\newblock In {\em Proc. IEEE Int. Conf. Acoust. Speech Signal Process.},
  volume~4, pages iv/553--iv/556, March 2005.

\bibitem{Han_Diao}
Chenzhe Diao and Bin Han.
\newblock Quasi-tight framelets with high vanishing moments derived from
  arbitrary refinable functions.
\newblock {\em Applied and Computational Harmonic Analysis}, 2018.

\bibitem{Hernandez_1996_4563}
E.~Hernandez and G.~Weiss.
\newblock {\em A first course on wavelets}.
\newblock CRC Press, Boca Raton, FL, 1996.

\bibitem{HAST201465}
Anders Hast.
\newblock Simple filter design for first and second order derivatives by a
  double filtering approach.
\newblock {\em Pattern Recognit. Lett.}, 42:65 -- 71, 2014.

\bibitem{GOH2008}
Say~Song Goh and K.~M. Teo.
\newblock Extension principles for tight wavelet frames of periodic functions.
\newblock {\em Appl. Comput. Harmon. Anal.}, 25(2):168 -- 186, 2008.

\bibitem{karSpie}
Nikolaos Atreas, Nikolaos Karantzas, Manos Papadakis, and Theodoros
  Stavropoulos.
\newblock Exploring neuronal synapses with directional and symmetric frame
  filters with small support.
\newblock In {\em Proc.SPIE}, volume 10394, pages 10394 -- 10394 -- 18, 2017.

\bibitem{Daub1}
Ingrid Daubechies.
\newblock Orthonormal bases of compactly supported wavelets.
\newblock {\em Comm. Pure Appl. Math.}, 41(7):909--996, 1988.

\end{thebibliography}
\end{document}